\documentclass{article}

\usepackage[a4paper]{geometry}
\usepackage{xspace}
\usepackage{amssymb}
\usepackage{amsthm}
\usepackage{amsmath}
\usepackage{graphicx}
\usepackage{times}
\usepackage{subcaption}

\newtheorem{claim}{Claim}
\newtheorem{theorem}{Theorem}
\newtheorem{lemma}{Lemma}

\newtheorem{property}{Property}

\graphicspath{{./figures/}}

\newcommand{\p}{\pi}
\newcommand{\ph}{\frac{\pi}{2}}
\newcommand{\pt}{\frac{3\pi}{2}}

\newcommand{\C}{\mathcal C}

\newcommand{\nh}{N_{\frac{\pi}{2}}}
\newcommand{\nt}{N_{\frac{3\pi}{2}}}
\newcommand{\nf}{N_{2\pi}}
\newcommand{\na}{N_{\alpha}}

\newcommand{\opvr}{OPVR\xspace}

\newcommand{\bG}{\overline{G}}

\begin{document}

\title{Ortho-polygon Visibility Representations\\of Embedded Graphs
\thanks{Research of EDG, WD, GL and FM  supported in part by the MIUR project AMANDA, prot. 2012C4E3KT\_001. NSERC funding is gratefully acknowledged for WE and SW.}}


\author{E. Di Giacomo\thanks{Universit{\`a} degli Studi di Perugia, Italy, \texttt{\{name.surname\}@unipg.it}} \and W. Didimo\footnotemark[1] \and W. S. Evans\thanks{University of British Columbia, Canada, \texttt{will@cs.ubc.ca}} \and G. Liotta\footnotemark[1] \and H. Meijer\thanks{University College Roosevelt, The Netherlands, \texttt{h.meijer@ucr.nl}} \and F. Montecchiani\footnotemark[1] \and S. K. Wismath\thanks{University of Lethbridge, Canada, \texttt{wismath@uleth.ca}}}

\date{}

\maketitle

\begin{abstract}
An ortho-polygon visibility representation of an $n$-vertex embedded graph $G$ (\opvr of $G$) is an embedding-preserving drawing of $G$ that maps every vertex  to a distinct orthogonal polygon and each edge to a vertical or horizontal visibility between its end-vertices. The vertex complexity of an \opvr of $G$ is the minimum $k$ such that every polygon has at most $k$ reflex corners. We present polynomial time algorithms that test whether $G$ has an \opvr and, if so, compute one of minimum vertex complexity. We argue that the existence and the vertex complexity of an \opvr of $G$ are related to its number of crossings per edge and to its connectivity. More precisely, we prove that if $G$ has at most one crossing per edge (i.e., $G$ is a 1-plane graph), an \opvr of $G$ always exists while this may not be the case if two crossings per edge are allowed. Also, if $G$ is a 3-connected 1-plane graph,  we can compute  an  \opvr of $G$ whose vertex complexity is bounded by a constant in $O(n)$ time. However, if $G$ is a 2-connected 1-plane graph, the vertex complexity of any \opvr of $G$  may be $\Omega(n)$. In contrast, we describe a family of 2-connected 1-plane graphs for which an embedding that guarantees constant vertex complexity can be computed in $O(n)$ time. Finally, we present the results of an experimental study on the vertex complexity of ortho-polygon visibility representations of 1-plane graphs.

\end{abstract}

\section{Introduction}

\emph{Visibility representations} are among the oldest and most studied methods to display graphs. The first papers appeared between the late 70s and the mid 80s, mostly motivated by VLSI applications (see, e.g.,~\cite{Duchet1983319,ov-grild-78,DBLP:journals/dcg/RosenstiehlT86,TamassiaTollis86,t-prg-84,DBLP:conf/compgeom/Wismath85})). These papers were devoted to \emph{bar visibility representations (BVR)} of planar graphs, where the vertices are modeled as non-overlapping horizontal segments, called \emph{bars}, and the edges correspond to vertical visibilities, i.e., vertical segments that do not intersect any bar other than at their end points.
The study of  visibility  representations of non-planar graphs started about ten years later when \emph{rectangle visibility representations} were introduced in the computational geometry and graph drawing communities (see, e.g.,~\cite{DBLP:journals/dam/DeanH97,DBLP:journals/comgeo/HutchinsonSV99,DBLP:journals/ipl/KantLTT97,DBLP:conf/cccg/Shermer96}). In a rectangle visibility representation every vertex is represented as an axis-aligned rectangle and two vertices are connected by an edge using either horizontal or vertical visibilities. Figure~\ref{fi:rvr} is an example of a rectangle visibility representation of the complete graph $K_5$. Rectangle visibility representations are an attractive way to draw a non-planar graph: Edges are easy to follow because they do not bend and can have only one of two possible slopes, edge crossings are perpendicular, textual labels associated with the vertices can be inserted in the rectangles.

Motivated by the NP-hardness of recognizing whether a graph admits a rectangle visibility representation~\cite{DBLP:conf/cccg/Shermer96}, Streinu and Whitesides~\cite{DBLP:conf/stacs/StreinuW03} initiated the study of rectangle visibility representations that must respect a set of topological constraints. They proved that if a graph $G$ is given together with the cyclic order of the edges around each vertex, the outer face, and a horizontal/vertical direction for each edge, then there exists a polynomial-time algorithm to test whether $G$ admits a rectangle visibility representation that respects these constraints. Biedl {\em et al.}~\cite{SoCG} have recently shown that testing the representability of $G$ is polynomial also with a different set of topological constraints, namely when $G$ is given with an embedding  that must be preserved in the rectangle visibility representation (an embedding specifies the cyclic order of the edges around each vertex and around each crossing, and the outer face). In these constrained settings, however, even structurally simple ``almost planar'' graphs may not admit a  rectangle visibility representation. For example, although the embedded graph of Figure~\ref{fi:flow} is 1-plane (i.e., it has at most one crossing per edge), it does not admit an embedding-preserving rectangle visibility representation~\cite{SoCG}.

\begin{figure}[tb]
    \centering
    \begin{minipage}[b]{.3\textwidth}
    	\centering
    	\includegraphics[scale=1,page=1]{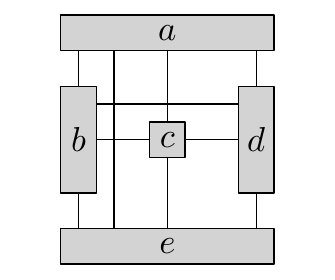}
    	\subcaption{}\label{fi:rvr}
    \end{minipage}
    \hfil
    \begin{minipage}[b]{.3\textwidth}
    	\centering
    	\includegraphics[scale=1,page=2]{figures/example2}
    	\subcaption{}\label{fi:flow}
    \end{minipage}
    \hfil
    \begin{minipage}[b]{.3\textwidth}
    	\centering
    	\includegraphics[scale=1,page=3]{figures/example2}
    	\subcaption{}\label{fi:opvr}
    \end{minipage}
    \caption{(a) A rectangle visibility representation of $K_5$. (b) An embedded graph $G$ that does not admit an embedding-preserving rectangle visibility representation. (c) An embedding-preserving \opvr of $G$ with vertex complexity one.}
\end{figure}

In this paper we introduce a generalization of rectangle visibility representations, we study to what extent such a generalization enlarges the family of graphs that are representable, and we describe testing and drawing algorithms. Let $G$ be an embedded graph. An \emph{ortho-polygon visibility representation} of  $G$ (\emph{\opvr} of $G$) is an embedding-preserving drawing of $G$ that maps each vertex to an orthogonal polygon, disjoint from the others, and each edge to a vertical or horizontal visibility between its end-vertices. For example, Figure~\ref{fi:opvr} is an embedding-preserving \opvr of the graph in Figure~\ref{fi:flow}.
In Figure~\ref{fi:opvr} all vertices except two are rectangles: The non-rectangular vertices have a reflex corner each; intuitively, each of them is ``away from a rectangle'' by one reflex corner. We say that the \opvr of Figure~\ref{fi:opvr} has \emph{vertex complexity} one. More  generally, we say that an \opvr has vertex complexity $k$, if $k$ is the maximum number of reflex corners over all vertex polygons in the representation.
We are not only interested in characterizing and testing what graphs admit an \opvr, but we also aim at computing representations of minimum vertex complexity (rectangle visibility representations if possible).
The main results in this paper can be listed as follows.

\begin{itemize}

\item In Section~\ref{se:characterization} we present a combinatorial characterization of the graphs that admit an embedding-preserving \opvr. The characterization leads to an $O(n^2)$-time algorithm that tests whether an embedded graph $G$ with $n$ vertices admits an embedding-preserving \opvr. If the test is affirmative, we also show that an embedding-preserving \opvr of $G$ with minimum vertex complexity can be computed in $O(n^\frac{5}{2} \log^\frac{3}{2}n)$ time. An implication of this characterization is that any 1-plane graph admits an embedding-preserving \opvr.

\item In Sections~\ref{se:background-edgepartitions} and~\ref{se:3conn-bounds} we prove that every 3-connected 1-plane graph admits an \opvr whose vertex complexity is bounded by a constant and that this representation can be computed in $O(n)$ time. This implies an $O(n^{\frac{7}{4}}\sqrt{\log n})$-time algorithm to compute {\opvr}s of minimum vertex complexity for these graphs. Biedl {\em et al.}~\cite{SoCG} proved that not every 3-connected 1-plane graph has a representation with zero vertex complexity, and in fact we also show a lower bound of two for infinitely many graphs of this family.

\item In Section~\ref{se:2conn-bounds} we study 2-connected 1-plane graphs. Note that not every 2-connected 1-plane graph can be augmented to become a 3-connected 1-plane graph, which has a strong impact on the vertex complexity of the corresponding {\opvr}s. Indeed, we prove that an embedding-preserving {\opvr} of a 2-connected 1-plane graph may require $\Omega(n)$ vertex complexity. On the positive side, for a special family of 2-connected 1-plane graphs we show that an embedding that guarantees constant vertex complexity can be computed in $O(n)$ time.

\item In Section~\ref{se:experiments} we present the results of an extensive experimental study on {\opvr}s of 1-plane graphs. This study aims at estimating both the vertex complexity of these drawings in practice and the percentage of vertices that are not represented as rectangles.
\end{itemize}

Section~\ref{se:preliminaries} contains preliminary definitions. In Section~\ref{se:tsm} we recall the basic ideas behind the Topology-Shape-Metrics framework, a key ingredient for the results presented throughout the paper. Conclusions and open problems are in Section~\ref{se:conclusions}.

We conclude this introduction by recalling that 1-planar graphs have been the subject of a rich literature in recent years. Particular attention has been given to   recognition and complexity problems (see, e.g.,~\cite{DBLP:conf/wads/BannisterCE13,DBLP:journals/siamcomp/CabelloM13,DBLP:journals/tcs/EadesHKLSS13,DBLP:journals/jgt/KorzhikM13}), straight-line drawings (see, e.g.,~\cite{DBLP:conf/gd/AlamBK13,t-rdg-JGT88}), right-angle crossing drawings (see, e.g.,~\cite{DBLP:conf/gd/DidimoL0M16,DBLP:journals/dam/EadesL13}), and visibility representations (see, e.g,~\cite{SoCG,DBLP:journals/jgaa/Brandenburg14,DBLP:journals/jgaa/Evans0LMW14}); see also~\cite{DBLP:journals/corr/KobourovLM17} for additional references and topics. In addition, two recent papers~\cite{DBLP:journals/tcs/EvansLM16,DBLP:journals/ipl/LiottaM16} study visibility representations of non-planar graphs where the edges are horizontal and vertical lines of sight and each vertex consists of two segments sharing an end-point. These representations can be turned into {\opvr}s of vertex complexity one by replacing the two segments of each vertex with an arbitrarily thin orthogonal polygon with one reflex corner.

\section{Preliminaries}\label{se:preliminaries}

A \emph{drawing} $\Gamma$ of a graph $G=(V,E)$ is a mapping of the vertices of $V$ to points of the plane, and of the edges in $E$ to Jordan arcs connecting their corresponding endpoints but not passing through any other vertex. We only consider \emph{simple} drawings, i.e., drawings such that two arcs representing two edges have at most one point in common, and this point is either a common endpoint or a common interior point where the two arcs properly cross each other. $\Gamma$ is \emph{planar} if no edge is crossed. A \emph{planar graph} is a graph that admits a planar drawing.

A planar drawing of a graph subdivides the plane into topologically connected regions, called \emph{faces}. The infinite region is the \emph{outer face}. A \emph{planar embedding} of a planar graph is an equivalence class of planar drawings that define the same set of faces. A \emph{plane graph} is a planar graph with a given planar embedding. Let $f$ be a face of a plane graph $G$. The number of vertices encountered in the closed walk along the boundary of $f$ is the \emph{degree} of $f$ and is denoted as $\deg(f)$. If $G$ is not 2-connected, a vertex may be encountered more than once, thus contributing  more than one unit to the degree of the face (see Figure~\ref{fi:face-degree}).  
The concept of a planar embedding is extended to non-planar drawings as follows. Given a non-planar drawing $\Gamma$, replace each crossing with a dummy vertex. The resulting planarized drawing has a planar embedding. An \emph{embedding} of a (non-planar) graph $G$ is an equivalence class of  drawings whose planarized versions have the same planar embedding.
An \emph{embedded graph} $G$ is a graph with a given embedding: An \emph{embedding-preserving} drawing $\Gamma$ of $G$ is a drawing of $G$ whose embedding coincides with that of $G$.

\begin{figure}[tb]
    \centering
    \begin{minipage}[b]{.3\textwidth}
    	\centering
    	\includegraphics[scale=1, page=1]{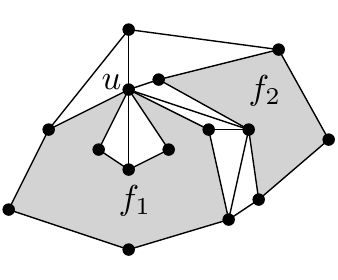}
    	\subcaption{}\label{fi:face-degree}
    \end{minipage}
    \hfil
    \begin{minipage}[b]{.3\textwidth}
    	\centering
    	\includegraphics[scale=1]{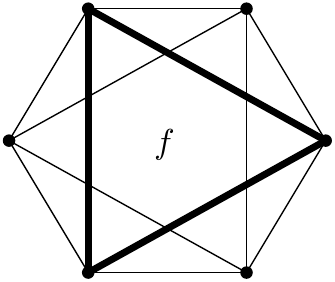}
    	\subcaption{}\label{fi:thickness2}
    \end{minipage}
    \caption{ (a) Face $f_1$ has degree 10 (since $u$ is counted twice), while face $f_2$ has degree 5. (b) An embedded graph that does not admit an embedding-preserving \opvr.}
\end{figure} 

A \emph{bar visibility representation (BVR)} of a plane graph $G$ maps the vertices of $G$ to non-overlapping horizontal segments, called \emph{bars}, and the edges of $G$ to vertical \emph{visibilities}.  A visibility is a vertical segment that does not intersect any bar other than those at its end-points. A BVR is \emph{strong} if each visibility between two bars corresponds to an edge of the graph, while it is \emph{weak} if visibilities between bars of non-adjacent vertices may occur. 

An \emph{orthogonal polygon} is a simple polygon whose sides are axis-aligned. A \emph{corner} of an orthogonal polygon is a point of the polygon where a horizontal and a vertical side meet. A corner is a \emph{reflex corner} if it forms a $\pt$ angle inside the polygon. An \emph{ortho-polygon visibility representation} (\emph{\opvr}) of a graph $G$ maps each vertex $v$ of $G$ to a distinct orthogonal polygon $P(v)$ and each edge $(u,v)$ of $G$ to a vertical or horizontal visibility connecting $P(u)$ and $P(v)$ and not intersecting any other polygon $P(w)$, for $w \not \in \{u,v\}$. The intersection points between visibilities and polygons are the \emph{attachment points}. As in many papers on visibility representations~\cite{DBLP:journals/ipl/KantLTT97,DBLP:conf/stacs/StreinuW03,TamassiaTollis86,DBLP:conf/compgeom/Wismath85}, we assume the $\varepsilon$-visibility model, where the segments representing the edges can be replaced by strips of non-zero width; this implies that an attaching point never coincides with a corner of a polygon. An \opvr is on an \emph{integer grid} if all its corners and attachment points have integer coordinates.
Given an \opvr, we can extract a drawing from it as follows. For each vertex $v$, place a point inside polygon $P(v)$ and connect it to all the attachment points on the boundary of $P(v)$; this can be done without creating any crossings and preserving the circular order of the edges around the vertices. Thus, we refer to an \opvr as a drawing and we extend all the definitions given for drawings to {\opvr}s. An \opvr $\gamma$ of an embedded graph is \emph{embedding-preserving} if the drawing extracted from $\gamma$ is embedding-preserving. 
When computing an \opvr we would like to use polygons that are not ``too complex'', ideally only rectangles. 
The \emph{vertex complexity} of an \opvr is the maximum number of reflex corners over all vertex polygons in the representation. An \emph{optimal \opvr} is an \opvr with minimum vertex complexity. In what follows, if this leads to no confusion, we shall use the term \emph{edge} to indicate both an edge and the corresponding visibility, and the term \emph{vertex} for both a vertex and the corresponding polygon.      

\section{The Topology-Shape-Metrics Framework}\label{se:tsm}

The \emph{topology-shape-metrics (TSM)} framework was introduced by Tamassia~\cite{t-eggmnb-87} to compute \emph{orthogonal drawings} of graphs (see also Chapter 5 in~\cite{dett-gdavg-99}). In an orthogonal drawing of a degree-4 graph each edge is a polyline of horizontal and vertical segments. A \emph{bend} is a point shared by two consecutive segments of an edge. An angle formed by two consecutive segments incident to the same vertex is a \emph{vertex-angle}; an angle at a bend is a \emph{bend-angle}. The following basic property holds~\cite{dett-gdavg-99}.

\begin{property}\label{pr:face}
Let $f$ be a face of an orthogonal drawing and let $\na(f)$ be the number of angles (vertex-angles and bend-angles) of value $\alpha$ inside $f$, with $\alpha \in \{\frac{\pi}{2},\frac{3\pi}{2}, 2\pi\}$. Then:
$\nh(f)-\nt(f)-2\nf(f) = 4$ if $f$ is an internal face and $\nh(f)-\nt(f)-2\nf(f) = - 4$ if $f$ is the outer face.
\end{property}

Given a degree-4 graph $G$, the TSM computes an orthogonal drawing $\Gamma$ of $G$ with a minimum number of bends. It works in three steps. The first step, called \emph{planarization}, computes an embedding of $G$ and replaces crossing points with dummy vertices. The resulting plane graph $G'$ has $n+c$ vertices, where $n$ and $c$ are the number of vertices and crossings of $G$, respectively. The second step, called \emph{orthogonalization}, computes an \emph{orthogonal representation} $H$ of $G'$, which specifies the values of all vertex-angles and the sequence of bend-angles along each edge. It defines an ``orthogonal shape'' of the final drawing, without specifying the length of the edge segments (a more precise definition of orthogonal representations can be found in Appendix~\ref{ap:orthorep-flow}). $H$ is computed by means of a flow network $N$, where each unit of flow corresponds to a $\ph$ angle. Each \emph{vertex-node} in $N$ corresponds to a vertex of $G'$ and  supplies 4 units of flow; each \emph{face-node} in $N$ corresponds to a face of $G'$ and demands an amount of flow proportional to its degree. Bends along edges correspond to units of flow transferred across adjacent faces of $G'$ through the corresponding arcs of $N$, and each bend has a unit cost in $N$ (more details can be found in Appendix~\ref{ap:orthorep-flow}). Network $N$ is constructed in $O(n+c)$ time since it has $O(n+c)$ nodes and arcs. Also, it always admits a feasible flow. A feasible flow $\Phi$ of cost $b$ of $N$ defines an orthogonal representation $H$ of $G'$ with $b$ bends, and \emph{vice versa}. The third step, called \emph{compaction}, computes an orthogonal drawing that preserves the shape defined by $H$, by assigning node and bend coordinates. It takes $O(n+c+b)$ time and the resulting drawing lies on an integer grid of size $O(n+c+b) \times O(n+c+b)$.

\section{Test and Optimization for Embedded Graphs}\label{se:characterization}

Any embedded graph $G$ that admits an \opvr is biplanar, i.e., its edge set can be bicolored so that each color class induces a planar subgraph (for example, color the horizontal edges of an \opvr of $G$ red and the vertical edges blue). However, a biplanar embedded graph $G$ may not have an embedding-preserving \opvr. An example is given in Figure~\ref{fi:thickness2} (thin and bold edges define the two color classes). The boundary of face $f$ in the figure contains six edge crossings and no vertices. In any \opvr of $G$, each crossing forms a $\ph$ angle inside $f$, thus the orthogonal polygon representing $f$ would have six $\ph$ corners and no $\pt$ corners in its interior, which is impossible.

In the following we first describe an algorithm that, given an embedded graph $G$ that admits an embedding-preserving \opvr, computes an optimal \opvr of $G$ (Lemma~\ref{le:opt}). Then, we describe a topological characterization of the embedded graphs that admit an embedding-preserving \opvr (Lemma~\ref{le:characterization}). This leads to an efficient testing algorithm and it implies that the embedded graphs with at most one crossing per edge, i.e., the \emph{1-plane graphs}, always admit an embedding-preserving \opvr. Both our results extend the topology-shape-metrics framework to handle {\opvr}s.

\begin{figure}[tb]
\centering
    \begin{minipage}[b]{.48\textwidth}
    	\centering
    	\includegraphics[width=1\textwidth, page=1]{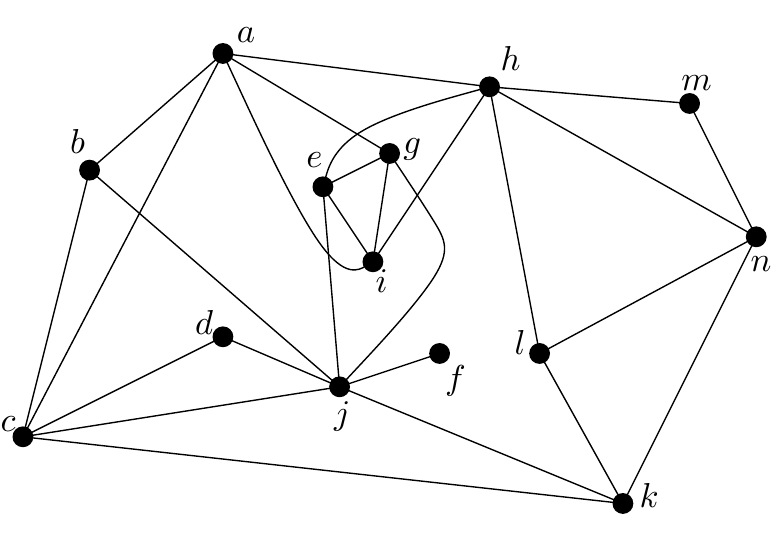}
    	\subcaption{$G$}\label{fi:characterization-1}
    \end{minipage}
    \begin{minipage}[b]{.48\textwidth}
    	\centering
    	\includegraphics[width=1\textwidth, page=3]{characterization}
    	\subcaption{$\bG$}\label{fi:characterization-3}
    \end{minipage}
    \begin{minipage}[b]{.48\textwidth}
    	\centering
    	\includegraphics[width=1\textwidth, page=5]{characterization}
    	\subcaption{$\gamma$}\label{fi:characterization-5}
    \end{minipage}
    \begin{minipage}[b]{.48\textwidth}
    	\centering
    	\includegraphics[width=1\textwidth, page=6]{characterization}
    	\subcaption{$\Gamma$}\label{fi:characterization-6}
    \end{minipage}
    \caption{(a) An embedded graph $G$ and (b) its planarized expansion $\bG$.  (c) An \opvr $\gamma$ of $G$ and (d) the orthogonal drawing $\Gamma$ obtained from $\gamma$.}
\end{figure}

\smallskip \noindent{\bf Our Approach.} To exploit the TSM framework, we define a new plane graph $\bG$ obtained from the input embedded graph $G$ as follows (refer to Figures~\ref{fi:characterization-1} and~\ref{fi:characterization-3}). Replace each vertex $v$ with a cycle $C(v)$ of $d=\deg(v)$ vertices, so that each of these vertices is incident to one of the edges formerly incident to $v$, preserving the circular order of the edges around $v$. If $d=1$ or $d=2$, $C(v)$ is a self-loop or a pair of parallel edges, respectively. $C(v)$ is the \emph{expansion cycle} of $v$; the vertices and the edges of $C(v)$ are the \emph{expansion vertices} and the \emph{expansion edges}, respectively. Also, replace crossings with \emph{dummy vertices}. $\bG$ is called the \emph{planarized expansion of $G$}. The edges of $\bG$ that are not expansion edges are the \emph{real edges}. Note that a real edge of $\bG$ corresponds either to an  uncrossed edge of $G$ or to a portion of a crossed edge of $G$. Clearly, each expansion vertex has degree 3 and each dummy vertex has degree 4.  The next lemma and properties immediately follow (see also Figs~\ref{fi:characterization-5} and~\ref{fi:characterization-6}).

\begin{lemma}\label{le:opvr-od}
An embedded graph $G$ admits an embedding-preserving \opvr if and only if $\bG$ admits an orthogonal representation with the following properties: {\bf P1.} Each vertex-angle inside an expansion cycle has value $\p$. {\bf P2.} Each real edge has no bend.
\end{lemma}
\begin{proof}
Let $\gamma$ be an embedding-preserving \opvr of $G$ (see, e.g., Figure~\ref{fi:characterization-5}). Replace each attachment point and each crossing point with a vertex (see Figure~\ref{fi:characterization-6}). The resulting drawing is a planar orthogonal drawing, whose orthogonal representation satisfies properties {\bf P1} and {\bf P2}.
In the other direction, assume that $\bG$ admits an orthogonal representation $H$ that satisfies {\bf P1} and {\bf P2}, and let $\Gamma$ be an orthogonal drawing with orthogonal representation $H$. Then to obtain an embedding-preserving \opvr of $G$, we replace each degree-4 vertex of $\Gamma$ by a crossing point, and every other vertex by an attachment point. In other words, each expansion cycle is replaced by the polygon representing it in $\Gamma$, and each edge of $G$ is represented by a visibility segment.
\end{proof}

\begin{property}\label{pr:bg-1}
If $G$ is biplanar, for each face $f$ of $\bG$ that is not an expansion cycle, $\deg(f) \geq 4$.
\end{property}
\begin{proof}
The faces of $\bG$ that are not expansion cycles arise from the faces of $G$.
Since $G$ is simple, every face $f$ of $G$ has degree at least three. If $\deg(f) \geq 4$, then $f$ clearly gives rise to a face of degree at least four in $\bG$. If $\deg(f)=3$ then $f$ cannot consist of crossing points only, otherwise there would be three mutually crossing edges and $G$ would not be biplanar. Hence, $f$ has at least one vertex on its boundary, and this vertex will correspond to two expansion vertices in $\bG$; then the face arising from $f$ in $\bG$ has degree at least four.
\end{proof}

\begin{property}\label{pr:bg-2}
If $G$ admits an embedding-preserving \opvr, then for every internal face $f$ of $\bG$ consisting only of dummy vertices, $\deg(f)=4$.
\end{property}
\begin{proof}
Suppose that $G$ has an embedding-preserving \opvr and that $f$ is an internal face of $\bG$ formed by dummy vertices only. By Lemma~\ref{le:opvr-od}, $\bG$ has an orthogonal representation with no bend on the edges of $f$, and all the vertex-angles inside $f$ have value $\ph$. Due to Property~\ref{pr:face} this implies that $\deg(f)=4$.
\end{proof}

\begin{lemma}\label{le:opt}
Let $G$ be an $n$-vertex embedded graph that admits an embedding-preserving \opvr. There exists an $O(n^{\frac{5}{2}}\log^{\frac{3}{2}} n)$-time algorithm that computes an embedding-preserving optimal \opvr $\gamma$ of $G$. Also, $\gamma$ has the minimum number of total reflex corners among all embedding-preserving optimal {\opvr}s of $G$.
\end{lemma}
\begin{proof}
Since $G$ admits an embedding-preserving \opvr, it is biplanar. Hence it has $m \leq 6n - 12$ edges.
By Lemma~\ref{le:opvr-od}, an \opvr of $G$ can be found by computing an orthogonal representation of $\bG$ that satisfies {\bf P1} and {\bf P2}. This can be done by computing a feasible flow in the Tamassia flow network $N$ associated with $\bG$, subject to the following constraints: $(i)$ Every arc of $N$ from a vertex-node to a face-node has fixed flow 2 if the face-node corresponds to an expansion cycle (which implies a $\p$ angle inside the cycle), and fixed flow 1 otherwise (which implies a $\ph$ angle inside the face); $(ii)$ Arcs between two face-nodes such that neither corresponds to an expansion cycle of $\bG$ are removed (to avoid bends on the real edges). A feasible flow for $N$ may not correspond to an optimal \opvr. To minimize the vertex complexity we construct a different flow network as follows. 

The amount of flow moved from a vertex-node to an adjacent face-node is fixed \emph{a priori}, and thus we can construct from $N$ an equivalent flow network $N'$, such that all vertex-nodes are removed and their supplies are transferred to the supply of the adjacent face-nodes. Specifically, each face-node $v_f$ corresponding to an expansion cycle $f$ receives $2\deg(f)$ units of flow, while its demand is $2\deg(f)-4$ by definition. This is equivalent to saying that $v_f$ will supply 4 units of flow in $N'$. Similarly, each face-node $v_f$ corresponding to a face $f$ that is not an expansion cycle receives $\deg(f)$ units of flow, while its demand is $2\deg(f)-4$ (or $2\deg(f)+4$ if $f$ is the outer face). This is equivalent to saying that $v_f$ will demand flow $\deg(f)-4$ ($\deg(f)+4$ if $f$ is the outer face) in $N'$. By Property~\ref{pr:bg-1}, $\deg(f) \geq 4$ and therefore $\deg(f)-4 \geq 0$. We now consider every face $f$ of $\bG$ having dummy vertices only (if any), and the corresponding face-node $v_f$ in $N'$. Note that $v_f$ is an isolated node of $N'$. Since $G$ admits an embedding-preserving \opvr, by Property~\ref{pr:bg-2},  $\deg(f) = 4$; hence, we can remove $v_f$ from $N'$ and conclude that $f$ must be drawn as a rectangle. Thus, every face-node in $N'$ corresponds to a face of $\bG$ with at least one expansion vertex on its boundary. Since every expansion vertex belongs to at most three faces of $\bG$ and there are $O(n)$ expansion vertices, $N'$ has $O(n)$ nodes and arcs.  

We also add gadgets to the network $N'$ in order to impose an upper bound $h$ on the number of reflex corners inside the polygons representing the expansion cycles. Let $v_f$ be a node of $N'$ corresponding to an expansion cycle $f$. We replace $v_f$ with two face-nodes: a node $v_f^{in}$, with zero supply and demand; and a node $v_f^{out}$, with the same supply as $v_f$ (which is 4). The incoming edges of $v_f$ become incoming edges of $v_f^{in}$, while the outgoing edges of $v_f$ become outgoing edges of $v_f^{out}$. Finally, we add an edge $(v_f^{in},v_f^{out})$ with capacity $h$. Let $N''$ be the flow network resulting by applying this transformation to all nodes of $N'$ corresponding to expansion cycles. Since each unit of flow entering in $v_f$ (now in $v_f^{in}$) corresponds to a $\pt$ angle inside $f$, a feasible flow of $N''$ defines an orthogonal representation where each expansion cycle is a polygon with at most $h$ reflex corners, i.e., such a feasible flow defines an \opvr having vertex complexity at most $h$. $N''$ is computed in $O(n)$ time and has $O(n)$ nodes and arcs, as $N'$. In order to guarantee that the \opvr has the minimum number of reflex corners among those with vertex complexity at most $h$, we compute a feasible flow of minimum cost. In particular, we apply the min-cost flow algorithm of Garg and Tamassia~\cite{Garg1997}, whose complexity is $O(\chi^{\frac{3}{4}}{m''}\sqrt{\log {n''}})$, where $n''$ and $m''$ are the number of nodes and arcs of $N''$, respectively, and $\chi$ is the cost of the flow\footnote{Note that we cannot use the faster min-cost flow algorithm in~\cite{DBLP:journals/jgaa/CornelsenK12} because $N''$ may not be planar (due to the gadgets introduced in order to transform $N'$ into $N''$).}. As already observed, both $n''$ and $m''$ are $O(n)$. Also, since the value of the flow is $O(n)$ and since in a min-cost flow each unit of flow moved along an augmenting path can traverse each face-node at most once, we have $\chi = O(n^2)$. Hence, a min-cost flow of $N''$ (if it exists) is computed in $O(n^{\frac{5}{2}}\sqrt{\log n})$ time. 

The supplied flow in $N''$ is $4n$ (four units for each expansion cycle) and each unit of a min-cost flow can traverse a face-node at most once. Thus, the vertex complexity of an embedding-preserving optimal \opvr of $G$ is $k \leq 4n$. We can find the value of $k$ by performing a binary search in the range $[0,4n]$, testing, for each considered value $h$, if an \opvr with vertex complexity at most $h$ exists. The number of tests is $O(\log n)$ and each test takes $O(n^{\frac{5}{2}}\sqrt{\log n})$ time, with the algorithm described above. Thus, computing an orthogonal representation $H$ corresponding to an \opvr with vertex complexity $k$ takes $O(n^{\frac{5}{2}}\log^{\frac{3}{2}}n)$ time. A drawing of $H$ is computed with the compaction step of the TSM. Since $H$ has at most $k \cdot n$ bends, this step can be executed in $O((k+1)n+c)=O(n^2)$ time.
\end{proof}

\begin{figure}[tb]
    \centering
    	\centering
    	\includegraphics[scale=0.5, page=4]{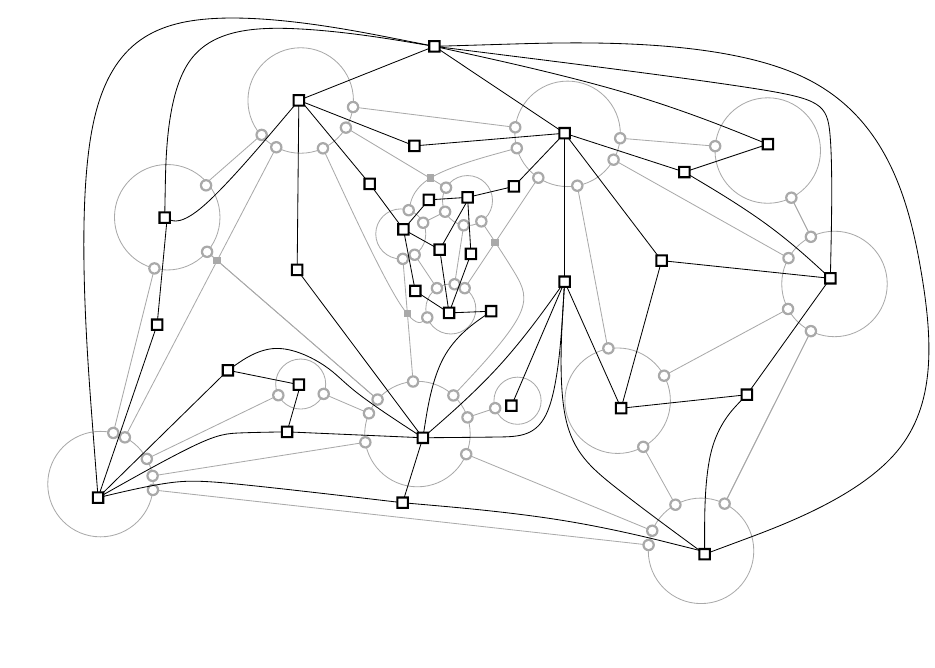}
    \caption{The simplified dual $\bG^*$ of the planarized expansion $\bG$ in Figure~\ref{fi:characterization-3}.\label{fi:characterization-4}}
\end{figure}

To describe our characterization, we introduce a new plane graph associated with the planarized expansion $\bG$ of $G$. Let $\bG^*$ be the dual graph of $\bG$ where the dual edges associated with the real edges are removed. $\bG^*$  has a vertex for each face of $\bG$ and an edge between two vertices for every edge of an expansion cycle shared by the two corresponding faces. We call $\bG^*$ the \emph{simplified dual} of $\bG$ (see also Figure~\ref{fi:characterization-4}). Given a connected component $\C$ of $\bG^*$, denote by $F_{\C}$ the set of faces of $\bG$ corresponding to the vertices of $\C$, by $F_{\C}^{ex}$ the subset of $F_{\C}$ corresponding to the expansion cycles, and by $F_{\C}^{nex}$ the set $F_{\C} \setminus F_{\C}^{ex}$. Finally, let $f_{out}$ be the outer face of $\bG$. We give the following characterization.

\begin{lemma}\label{le:characterization} 
An embedded graph $G$ admits an embedding-preserving \opvr if and only if for each connected component $\C$ of $\bG^*$ we have:

\begin{equation}\label{eq:char}
\sum_{f \in F_{\C}^{nex}}\deg(f)=
    \begin{cases}
    4|F_\C| & \text{if $f_{out} \not \in F_\C$} \\
    4|F_\C|-8 & \text{if $f_{out} \in F_\C$}
    \end{cases}
\end{equation}
\end{lemma}

\begin{proof}
Suppose first that $G$ admits an embedding-preserving \opvr $\gamma$. By Lemma~\ref{le:opvr-od}, $\bG$ admits an orthogonal representation $H$ such that properties {\bf P1} and {\bf P2} hold. Consider any orthogonal drawing $\Gamma$ of $\bG$ that can be obtained from $H$, and let $\C$ be a connected component of $\bG^*$. For each face $f \in F_\C$ Property~\ref{pr:face} holds in $\Gamma$, and since there is no angle of $2\pi$ it follows that $\nh(f)-\nt(f)=4$. Summing over all faces we obtain

\begin{equation}\label{eq:sun-faces}
\sum_{f \in F_\C}(\nh(f)-\nt(f))=
\begin{cases}
    4|F_\C| & \text{if $f_{out} \not \in F_\C$} \\
    4|F_\C|-8 & \text{if $f_{out} \in F_\C$}
\end{cases}
\end{equation}

We can rewrite the left-hand side of Equation~\ref{eq:sun-faces} as follows: 
$\sum_{f \in F_\C}(\nh(f)$ - $\nt(f))$ = $\sum_{f \in F^{ex}_\C}(\nh^v(f)$ + $\nh^b(f)$ - $\nt^v(f)$ - $\nt^b(f))$ + $\sum_{f \in F^{nex}_\C}(\nh^v(f) + \nh^b(f) - \nt^v(f) - \nt^b(f))$, where the superscripts $b$ and $v$ indicate whether the angle is a bend-angle or a vertex-angle, respectively. Since we are using the $\varepsilon$-visibility model, the attachment points of the edges incident to the polygons $P(v)$ are not corners of $P(v)$. Thus, all corners of $P(v)$ are bends along the edges of $C(v)$, which means that $\nh^v(f)=\nt^v(f)=0$ for all faces $f \in F^{ex}_\C$. For the faces $f \in F^{nex}_\C$, we have instead that each vertex-angle of $f$ is a $\ph$ angle (i.e., $\nt^v(f)=0$). More precisely, if the vertex is a dummy vertex then it has degree four and therefore the four angles around it all measure $\ph$. If the vertex is an expansion vertex, then it forms a $\pi$ angle inside its expansion cycle (because the attaching points are not corners). Since the vertex has degree three, the other two angles around it measure $\ph$. The only edges that can have bends are the expansion edges because the real edges are (parts of) the edges of $G$ that are drawn as straight-line segments in $\gamma$. Since the expansion edges are shared by a face of $F^{ex}_\C$ and a face of $F^{nex}_\C$, each bend forming a $\ph$ angle inside a face of $F^{ex}_\C$ forms an angle of $\pt$ inside a face of $F^{nex}_\C$, and vice versa. This means that
$\sum_{f \in F^{ex}_\C}(\nh^b(f)-\nt^b(f))+\sum_{f \in F^{nex}_\C}(\nh^b(f)-\nt^b(f))=0$ and therefore $\sum_{f \in F_\C}(\nh(f)-\nt(f))=\sum_{f \in F^{nex}_\C}\nh^v(f)=\sum_{f \in F^{nex}_\C}\deg(f)$. Thus, Equation~\ref{eq:sun-faces} becomes Equation~\ref{eq:char}.

We now prove that if Equation~\ref{eq:char} holds for every connected component $\C$ of $\bG^*$, then $G$ admits an embedding-preserving \opvr. Consider the flow network $N'$ defined in the proof of Lemma~\ref{le:opt}. To prove the claim we show that if  Equation~\ref{eq:char} holds for every connected component $\C$ of $\bG^*$, then $N'$ admits a feasible flow. To this aim we observe that $N'$ and $\bG^*$ share the same set of vertices. That is, for each node $v_f$ of $\bG^*$ corresponding to a face $f$ in $F_{\C}$, there is a corresponding node $v_f$ in $N'$. Also, for each edge $(u,v)$ of $\bG^*$ there are two arcs, $(u,v)$ and $(v,u)$, in $N'$. It follows, that for every connected component $\C$ of $\bG^*$, there is a corresponding strongly connected component $\C'$ in $N'$.

The flow network $N'$ admits a feasible flow if and only if every connected component of $N'$ admits a feasible flow. Consider a connected component $\C'$ of $N'$. Since the capacities of the arcs of $N'$ are unbounded and $\C'$ is in fact strongly connected, a feasible flow of $\C'$ exists if and only if the total supply is equal to the total demand (see, e.g.,~\cite{ht-fffscn-08}). We now show that this is the case. Suppose first that $f_{out} \not \in F_\C$. The total supply of the nodes of $\C'$ is $4 |F_\C^{ex}|$, while the total demand is $\sum_{f \in F_{\C}^{nex}}(\deg(f)-4)$. By Equation~\ref{eq:char} the total demand can be written as $4|F_{\C}|-4|F_{\C}^{nex}|$, which is clearly equal to the total supply $4 |F_\C^{ex}|$. If $f_{out} \in F_\C$, the total supply is again $4 |F_\C^{ex}|$, while the total demand is $\sum_{f \in F_{\C}^{nex}}(\deg(f)-4)+8$. Again, by Equation~\ref{eq:char} the total demand can be written as $4|F_\C^{nex}|-8-4|F_{\C}|+8=4|F_{\C}|-4|F_{\C}^{nex}|$, which is equal to the total supply $4 |F_\C^{ex}|$. This concludes the proof that $N'$ admits a feasible flow, which implies the existence of an embedding-preserving \opvr of $G$.
\end{proof}

The characterization of Lemma~\ref{le:characterization} immediately leads to an $O(n+c)$-time algorithm that tests whether an embedded graph $G$ with $n$ vertices and $c$ crossings admits an embedding-preserving \opvr. Indeed, the size of $\bG^*$ is $O(n+c)$ and the condition of Lemma~\ref{le:characterization} can be checked in linear time in the size of $\bG^*$. If $G$ is biplanar, it has at most $6n-12$ edges and $O(n+c) = O(n^2)$.
The next theorem summarizes the contribution of this section.

\begin{theorem}\label{th:test-opt}
Let $G$ be an $n$-vertex embedded graph. There exists an $O(n^2)$-time algorithm that tests whether $G$ admits an embedding-preserving \opvr and, if so, it computes an embedding-preserving optimal \opvr $\gamma$ of $G$ in $O(n^{\frac{5}{2}}\log^{\frac{3}{2}}n)$ time.  Also, $\gamma$ has the minimum number of total reflex corners among all embedding-preserving optimal {\opvr}s of $G$.
\end{theorem}

It may be worth remarking that an alternative algorithm to test whether $G$ admits an embedding-preserving \opvr can be derived from the result in~\cite{DBLP:conf/sofsem/AlamKM16}. Alam {\em et al.}~\cite{DBLP:conf/sofsem/AlamKM16} showed an  algorithm to test whether an $n$-vertex biconnected plane graph $G$ admits an orthogonal drawing such that edges have no bends, and each face $f$ has most $k_f$ reflex corners. The time complexity of this algorithm is $O((nk)^{\frac{3}{2}})$-time, where $k = \max_{f \in G}{k_f}$. Thus, one can compute $\bG$ and split each expansion edge of $\bG$ with $4n$ subdivision vertices (the maximum number of reflex corners that a face can have). The resulting graph $\bG'$ has $O(n^2)$ vertices. Then one can apply the algorithm by Alam {\em et al.} on $\bG'$ with $k_f = 4n$ for every face $f$ of $G$.  However, this would lead to a time complexity  $O(n^\frac{9}{2})$.

We conclude this section by observing that the number of crossings per edge is a critical parameter for the ortho-polygon representability of an embedded graph, namely even two crossings per edge may give rise to a graph that cannot be represented -- see Figure~\ref{fi:thickness2}. On the positive side, the following theorem can be proved by applying Lemma~\ref{le:characterization}.

\begin{theorem}\label{th:1-planar}
Every 1-plane graph admits an embedding-preserving \opvr.
\end{theorem}
\begin{proof}
Let $G$ be a 1-plane graph with $n$ vertices, $m$ edges, and $c$ crossings.
Let $\bG$ be the planarized expansion of $G$, and let $\bG^*$ be the simplified dual of $\bG$. We first observe that $\bG^*$ is connected. If not  there would be two sets of faces of $G$ such that for a face $f_1$ from one set and a face $f_2$ from the other set, $f_1$ and $f_2$ do not share an edge of an expansion cycle. In other words, there exists  a cycle of $\bG$ that contains only dummy vertices. Such a cycle however can exist only if each of its edges has two dummy end-vertices, which is impossible because $G$ is 1-plane. We now show that $G$ satisfies the condition of Lemma~\ref{le:characterization}.

Since $\bG^*$ consists of one connected component, it contains the outer face and by Lemma~\ref{le:characterization} we have $\sum_{f \in F^{nex}}\deg(f)=4|F|-8$, where $|F|$ is the number of faces of $\bG$.  Each expansion vertex forms two angles inside faces that are not expansion cycles, while each dummy vertex forms four angles inside faces that are not expansion cycles. Hence, $\sum_{f \in F^{nex}}\deg(f)=2n_e+4n_d$, where $n_e$ and $n_d$ are the total number of expansion vertices and dummy vertices, respectively. We have that $n_e=2m$, $n_d=c$, and $|F|=n+f'$,  where $f'$ is the number of faces in the planarization of $G$.
 It follows that $2n_e+4n_d = 4m +4c$ and the condition of Lemma~\ref{le:characterization} becomes $4m+4c=4f'+4n-8$, that is $n+f'=m+c+2$. Let $n'$ and $m'$ be the number of vertices and edges, respectively, in the planarization of $G$. By  Euler's formula, we have $n'+f'=m'+2$. Since $n'=n+c$ and $m'=m+2c$, it follows that $n+c+f'=m+2c+2$ and therefore $n+f'=m+c+2$, which proves that $G$ satisfies Equation~\ref{eq:char}.
\end{proof}

Motivated by Theorem~\ref{th:1-planar}, we devote the next sections to the study of upper and lower bounds on the vertex complexity of 1-plane graphs.

\section{Types of Crossings and Edge Partitions of 1-plane Graphs}\label{se:background-edgepartitions}

In this section, we first classify  different types of crossings that arise in 1-plane graphs (Section~\ref{sse:background}).  
We then present a result about partitioning the edges of a 3-connected 1-plane graph so that each partition set induces a plane graph and one of these plane graphs has maximum vertex degree six, which is a tight bound (Section~\ref{sse:partition}). This result may be of independent interest since it contributes to recent combinatorial studies about partitioning the edge set of 1-plane graphs into two plane subgraphs having special properties (see, e.g.,~\cite{DBLP:journals/dam/Ackerman14,DBLP:journals/combinatorics/CzapH13,Lenhart201759}).
Moreover, the results in this section are then used to prove an upper bound of 12 and a lower bound of 2 on the vertex complexity of 3-connected 1-plane graphs (Section~\ref{se:3conn-bounds}). 

\subsection{Types of crossings in 1-plane graphs}\label{sse:background}

\begin{figure}[tb]
    \centering
    \begin{minipage}[b]{.18\textwidth}
    	\centering
    	\includegraphics[scale=0.9,page=1]{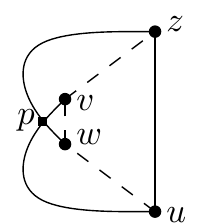}
    	\subcaption{B-conf.}\label{fi:bconf}
    \end{minipage}
    \begin{minipage}[b]{.18\textwidth}
    	\centering
    	\includegraphics[scale=0.9,page=2]{bconf-withcrossings}
    	\subcaption{Kite}\label{fi:kite}
    \end{minipage}
    \begin{minipage}[b]{.18\textwidth}
    	\centering
    	\includegraphics[scale=0.9,page=4]{bconf-withcrossings}
    	\subcaption{W-conf.}\label{fi:wconf}
    \end{minipage}
    \begin{minipage}[b]{.21\textwidth}
    	\centering
    	\includegraphics[scale=0.9,page=5]{bconf-withcrossings}
    	\subcaption{T-conf.}\label{fi:tconf}
    \end{minipage}
    \begin{minipage}[b]{.21\textwidth}
    	\centering
    	\includegraphics[scale=0.9,page=6]{bconf-withcrossings}
    	\subcaption{Aug. T-conf.}\label{fi:augtconf}
    \end{minipage}
    \caption{Crossing configurations of 1-plane graphs.\label{fi:crossings} }
\end{figure}

Let $G$ be a 1-plane graph and let $(u,v)$ and $(w,z)$ be two edges of $G$ that cross at a point $p$.  Edges $(u,v)$ and $(w,z)$ form a {\em B-configuration} if there exists an edge between their endpoints, say edge $(u,z)$, such that vertices $v$ and $w$ are inside the triangle $\{u,z,p\}$ (see  Figure~\ref{fi:bconf}, without dashed edges). If $(u,v)$ and $(w,z)$ form a B-configuration and all the edges of the 4-cycle $\{(u,w), (w,v), (v,z), (z,u)\}$ exist, then $(u,v)$ and $(w,z)$ form an \emph{augmented B-configuration} (see Figure~\ref{fi:bconf}, including the dashed edges). Two crossing edges $(u,v)$ and $(w,z)$ form a \emph{kite} in $G$, if the 4-cycle $\{(u,w), (w,v), (v,z), (z,u)\}$ exists, and the crossing point $p$ between $(u,v)$ and $(w,z)$ lies inside such a 4-cycle (see Figure~\ref{fi:kite}). Let $(u,v)$ and $(w,z)$ be two edges of $G$ that cross at a point $p$, and let $(u,x)$ and $(z,y)$ be two further edges of $G$ that cross at a point $q$. The four edges form a \emph{W-configuration} if vertices $v,w,x,y$ lie inside the cycle formed by the edge parts $(u,p)$, $(p,z)$, $(z,q)$, and $(q,u)$ (see Figure~\ref{fi:wconf}). B- and W-configurations were introduced by Thomassen~\cite{t-rdg-JGT88} to characterize the 1-plane graphs that admit an embedding-preserving straight-line drawing. In~\cite{SoCG}, Biedl {\em et al.} introduced an additional configuration called \emph{T-configuration}. They proved that a 1-plane graph admits an embedding-preserving rectangle visibility representation if and only if it does not contain B-, W-, or T-configurations~\cite{SoCG}. Let $(u,v)$ and $(w,z)$ be two edges crossing at a point $p$, let  $(u,y)$ and $(x,z')$ be two edges crossing at a point $q$, and let $(x,v')$ and $(w,y')$ be two edges crossing at a point $t$. If vertices $v,v',y,y',z,z'$ are inside the cycle formed by the edge parts $(u,q)$, $(q,x)$, $(x,t)$, $(t,w)$, $(w,p)$, and $(p,u)$, then the above six edges form a T-configuration (see Figure~\ref{fi:tconf}). Moreover, if $v=v'$, $y=y'$, $z=z'$ and $v,y,z$ form a triangle in $G$, then the above six edges form an \emph{augmented T-configuration} (see Figure~\ref{fi:augtconf}).

In the following sections we shall often refer to \emph{crossing-augmented 1-plane graphs}. A 1-plane graph $G$ is crossing-augmented, when for each pair of crossing edges $(u,v)$ and $(w,z)$, the subgraph of $G$ induced by $\{u,v,w,z\}$ is a $K_4$. We call  the four edges of the $K_4$ different from $(u,v)$ and $(w,z)$ \emph{cycle edges} of $(u,v)$ and $(w,z)$  -- they form a 4-cycle. Note that a 1-plane graph can always be made crossing-augmented in $O(n)$ time, by adding the missing cycle edges without introducing any new edge crossings (see, e.g.,~\cite{DBLP:conf/gd/AlamBK13,Suzuki2010}).

\subsection{Edge partition of 3-connected 1-plane graphs}\label{sse:partition}

An \emph{edge partition}  of a 1-plane graph $G$ is a coloring of each of its edges with one of two colors, \emph{red} and \emph{blue}, such that both the red graph $G_R$ induced by the red edges and the blue graph $G_B$ induced by the blue edges are plane graphs. 
The planar embedding of $G_R$ ($G_B$) is induced by the 1-planar embedding of $G$ when considering only the red (blue) edges. 
It is known that $G$ admits an edge partition such that $G_R$ is a forest~\cite{DBLP:journals/dam/Ackerman14,DBLP:journals/combinatorics/CzapH13}, and that if $G$  has $4n-8$ edges, then an edge partition such that $G_R$ has maximum vertex degree four always exists~\cite{Lenhart201759}. 
The following result, besides being of independent interest for the theory of 1-planarity, will be used to establish an upper bound on the vertex complexity of {\opvr}s of 3-connected 1-plane graphs.

\begin{theorem}\label{th:edge-decomposition}
Let $G$ be a 3-connected 1-plane graph with $n$ vertices.
There is an edge partition of $G$ such that the red graph has maximum vertex degree six and this bound is worst case optimal. Also, such an edge partition can be computed in $O(n)$ time.
\end{theorem}
\begin{proof}
We assume that $G$ is crossing-augmented. The proof relies on claims that describe properties of the cycle edges of $G$ which  make it possible to construct the desired partition of the edges of $G$. It is important to observe that, due to 1-planarity, a cycle edge  is not crossed in the subgraph induced by the four end-vertices of its two crossing edges. However, a cycle edge can be crossed in $G$, but, as shown in the next claim, no two cycle edges cross one another in $G$.

\begin{claim}\label{cl:cycleedges}
There are no two cycle edges of $G$ that cross each other.
\end{claim}
\begin{proof}
Refer to Figure~\ref{fi:cycledges}. Let $(u,v)$ and $(w,z)$ be two edges crossing each other, and assume for a contradiction that they are both cycle edges. This implies the existence of two pairs of crossing edges: $(u,x)$ and $(v,y)$, crossing at point $p$; $(w,x')$ and $(z,y')$, crossing at a point $q$. Then either $w$ or $z$ is inside the cycle $C$ composed of the following (parts of) edges: $(u,v)$, $(u,p)$, and $(v,p)$. Without loss of generality suppose that $w$ is inside $C$ and hence $z$ is outside $C$. It is immediate to see that either $(w,x')$ or $(z,y')$ crosses one among $(u,v)$, $(u,x)$ and $(v,y)$, and hence there is at least one edge crossed twice, which contradicts the 1-planarity of $G$.
\end{proof}

\begin{figure}[tb]
    \centering
    \begin{minipage}[b]{.24\textwidth}
    	\centering
    	\includegraphics[scale=1, page=1]{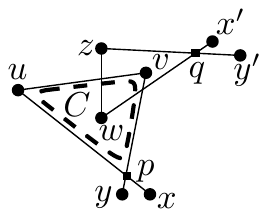}
    	\subcaption{}\label{fi:cycledges}
    \end{minipage}
    \begin{minipage}[b]{.24\textwidth}
    	\centering
    	\includegraphics[scale=1, page=2]{cycleedges}
    	\subcaption{}\label{fi:cycledges-2}
    \end{minipage}
    \begin{minipage}[b]{.24\textwidth}
    	\centering
    	\includegraphics[scale=1,page=1]{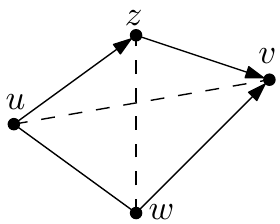}
    	\subcaption{}\label{fi:opposite-a}
    \end{minipage}
    \begin{minipage}[b]{.24\textwidth}
    	\centering
    	\includegraphics[scale=1,page=2]{opposite}
    	\subcaption{}\label{fi:opposite-b}
    	\end{minipage}
    \caption{(a) Illustrations for the proofs of (a) Claim~\ref{cl:cycleedges}; (b) Claim~\ref{cl:cycleedges-2}; (c)-(d) Claim~\ref{cl:opposite}.}
\end{figure}

\begin{claim}\label{cl:cycleedges-2}
Every edge of $G$ is the cycle edge of at most two pairs of crossing edges.
\end{claim}
\begin{proof}
Refer to Figure~\ref{fi:cycledges-2}. Let $(u,v)$ be a cycle edge shared by two pairs of crossing edges. These two pairs of crossing edges define a cycle $C$ (dashed in Figure~\ref{fi:cycledges-2}) such that no vertex inside $C$ can be connected with a vertex outside $C$, except through a path that contains $u$ or $v$. Suppose, for a contradiction, that there is a third pair of crossing edges having $(u,v)$ as a cycle edge. Then, for every two pairs of crossing edges among these three, there is a cycle $C_i$ $(i=1,2,3)$ passing through $u$ and $v$ and with the same property as $C$, that is, any path from a vertex inside $C_i$ to a vertex outside $C_i$ contains $u$ or $v$. Also, in any 1-planar embedding of $G$, one of these three cycles is such that the end-vertices of one of the three pairs of crossing edges are inside this cycle, and the end-vertices of another pair are outside it.  This implies that $u$ and $v$ are a separation pair, a contradiction with the fact that $G$ is 3-connected.
\end{proof}

Let $G_p$ be the plane graph obtained from $G$ by removing an edge for each pair of crossing edges. We can arbitrarily choose what edges to remove, provided that we never remove a cycle edge. Claim~\ref{cl:cycleedges} ensures that this choice is always feasible.  Let $G_p^+$ be a plane graph obtained by edge-augmenting $G_p$ so to become a plane triangulation.

We apply a {\em Schnyder trees decomposition}~\cite{DBLP:conf/soda/Schnyder90} to $G_p^+$. Schnyder~\cite{DBLP:conf/soda/Schnyder90} proved that the internal edges of a plane triangulation can be oriented such that each internal vertex has exactly three outgoing edges and the vertices of the outer face have no outgoing edge. We arbitrarily orient the edges of the outer face of $G_p^+$ and we obtain a {\em 3-orientation} of $G_p^+$, that is an orientation of its edges such that every vertex has at most three outgoing edges. Based on this 3-orientation, the following claim can be proved.

\begin{claim}\label{cl:opposite}
Let $(u,v)$ and $(w,z)$ be a pair of crossing edges of $G$. Then both $\{u,v\}$ or both $\{w,z\}$ have an outgoing edge in $G_p^+$ that is a cycle edge of $(u,v)$ and $(w,z)$.
\end{claim}
\begin{proof}
Consider the 4-cycle in $G_p^+$ formed by the four cycle edges of $(u,v)$ and $(w,z)$. Recall that these four edges are all present in $G_p^+$, since we did not remove any cycle edge. Suppose that $u$ has an outgoing edge, as shown in Figure~\ref{fi:opposite-a}. Then either $v$ has an outgoing edge, or both the edges $(z,v)$ and $(w,v)$ are oriented towards $v$. In both cases the statement holds. Suppose otherwise that both edges of $u$ are incoming. Then both $w$ and $z$ have an outgoing edge towards $u$, as shown in Figure~\ref{fi:opposite-b}.
\end{proof}

We use Claim~\ref{cl:opposite} to partition the edge set of $G$ as follows. For each pair of crossing edges $(u,v)$ and $(w,z)$ of $G$ we color red the edge connecting the pair $\{u,v\}$ or $\{w,z\}$ for which Claim~\ref{cl:opposite} holds. By this choice, each end-vertex of a red edge has one outgoing edge among the cycle edges of $(u,v)$ and $(w,z)$. Since every vertex is incident to at most three outgoing edges in $G_p^+$, and since each edge is the cycle edge of at most two pairs of crossing edges (Claim~\ref{cl:cycleedges-2}), by this procedure at most six edges for each vertex get the red color. 

The linear time complexity is a consequence of the fact that a 1-plane graph has $O(n)$ edges~\cite{Suzuki2010} and that Schnyder trees can be constructed in $O(n)$ time~\cite{DBLP:conf/soda/Schnyder90}.

\begin{figure}[t]
    \centering
    \includegraphics[scale=1,page=1]{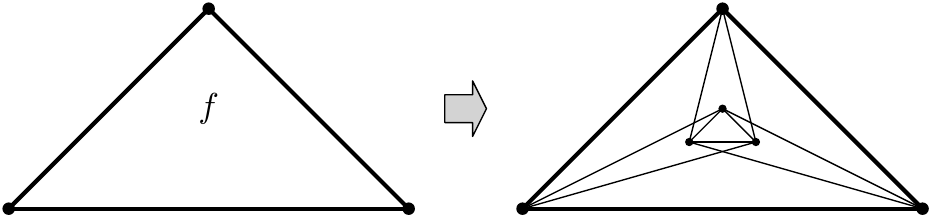}
    \caption{Illustration for the proof of Theorem~\ref{th:edge-decomposition}.}\label{fi:3connlb-1}
\end{figure}

We conclude the proof by showing that there exist 3-connected 1-plane graphs such that the red graph of any edge partition has maximum vertex degree at least $6$.
Let $G_p$ be a maximal plane graph with $n \geq 13$ vertices. Construct the graph $G$ from $G_p$ as follows. For each face $f$ of $G_p$ identify the three vertices of $f$ with the three vertices on the outer face of an augmented T-configuration (see Figure~\ref{fi:augtconf}). An illustration of the insertion of an augmented T-configuration inside a face $f$ is shown in Figure~\ref{fi:3connlb-1}. Graph $G$ is 3-connected and 1-plane by construction. Consider an edge partition of $G$. For every face $f$ of $G_p$, there are exactly three red edges. Each of these three red edges is incident to a vertex of $G_p$. Since $G_p$ has $2n-4$ faces, there are $3(2n-4)=6n-12$ red edges, each incident to a vertex of $G_p$. If the maximum vertex degree of the red graph is $k$, then it must be $kn \geq 6n - 12$, which implies $k \geq  6 - \frac{12}{n}$, and, since $k$ is integer, $k \geq 6$ for $n \geq 13$.
\end{proof}

\section{Vertex Complexity Bounds for 3-connected 1-plane Graphs}\label{se:3conn-bounds}

The edge partition of Theorem~\ref{th:edge-decomposition} can be used to construct an \opvr of a 3-connected 1-plane graph whose vertex complexity does not depend on the input size. We first describe the high-level idea behind this construction (see also Figure~\ref{fi:example3conn1plane} for an illustration) and then give a formal proof (Theorem~\ref{th:decomp-opvr}). 

Let $G$ be a 3-connected 1-plane graph with $n$ vertices, and let $G_B$ and $G_R$ be the plane graphs defined by the edge partition of Theorem~\ref{th:edge-decomposition}; see for example Figure~\ref{fi:example3conn1plane-1}. We first augment $G_B$ to a maximal plane graph (if needed), and then construct a BVR $\gamma_B$; see for example Figure~\ref{fi:example3conn1plane-2}. Assume that two vertices $u$ and $v$ are connected by a red edge and let $\gamma_B(u)$ and $\gamma_B(v)$ be the horizontal bars representing vertices $u$ and $v$ in $\gamma_B$, respectively. We attach a vertical bar to $\gamma_B(u)$ and a vertical bar to $\gamma_B(v)$ such that each vertical bar shares an endpoint with the horizontal bar and the two vertical bars can see each other horizontally. This makes it possible to draw the horizontal red edge $(u,v)$; see for example Figure~\ref{fi:example3conn1plane-3}.  Once all red edges have been added to $\gamma_B$, every vertex $v$ that has some incident red edge is represented as a ``rake''-shaped object consisting of one horizontal bar and at most six vertical bars (we have a vertical bar for each red edge incident to $v$ and there are at most six such edges). This ``rake''-shaped object can then be used as the skeleton of an orthogonal polygon that has two reflex corners per vertical bar; see for example Figure~\ref{fi:example3conn1plane-4}.

\begin{figure}[tb]
    \centering
    \begin{minipage}[b]{.24\textwidth}
    	\centering
    	\includegraphics[width=\textwidth,page=1]{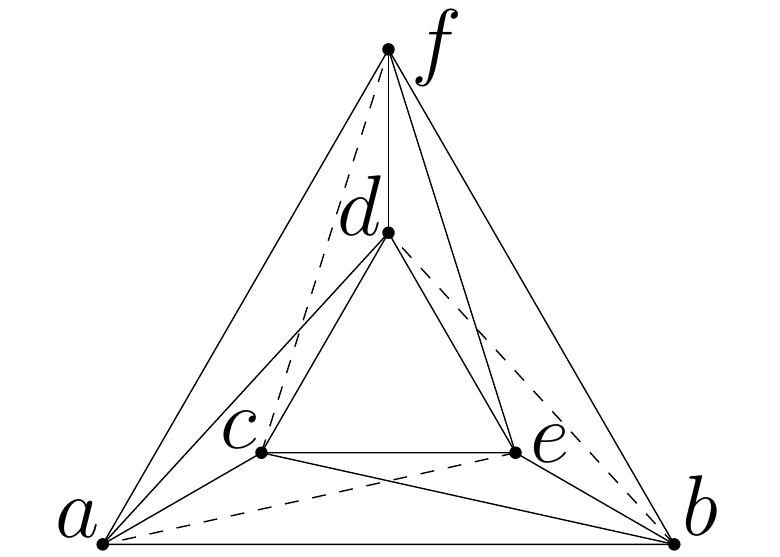}
    	\subcaption{}\label{fi:example3conn1plane-1}
    \end{minipage}
    \begin{minipage}[b]{.24\textwidth}
    	\centering
    	\includegraphics[width=\textwidth,page=2]{example3conn1plane}
    	\subcaption{}\label{fi:example3conn1plane-2}
    \end{minipage}
    \begin{minipage}[b]{.24\textwidth}
    	\centering
    	\includegraphics[width=\textwidth,page=3]{example3conn1plane}
    	\subcaption{}\label{fi:example3conn1plane-3}
    \end{minipage}
    \begin{minipage}[b]{.24\textwidth}
    	\centering
    	\includegraphics[width=\textwidth,page=4]{example3conn1plane}
    	\subcaption{}\label{fi:example3conn1plane-4}
    \end{minipage}
    \caption{(a) An edge partition of a 3-connected 1-plane graph $G$; red (blue) edges are dashed (solid); (b) A BVR $\gamma_B$ of $G_B$; (c) Insertion of the red edges into $\gamma_B$; (d) An \opvr of $G$.\label{fi:example3conn1plane}}
\end{figure}

We start with some additional definitions. A plane acyclic directed graph $G_p$ such that $G_p$ has a single source $s$ and a single sink $t$ that are both on the outer face, is called a \emph{plane $st$-graph} (see, e.g.,~\cite{DBLP:journals/dcg/RosenstiehlT86,TamassiaTollis86}).  For each vertex $v$ of a plane $st$-graph $G_p$, the incoming edges appear consecutively around $v$, as do the outgoing edges. Vertex $s$ only has outgoing edges, while vertex $t$ only has incoming edges (this particular transversal structure is known as a \emph{bipolar orientation}~\cite{DBLP:journals/dcg/RosenstiehlT86,TamassiaTollis86}). Each cycle $C$ of $G_p$ is bounded by two directed paths with a common \emph{origin} and \emph{destination}, called the \emph{left path} and \emph{right path} of $C$. 

Let $G$ be a 3-connected 1-plane graph with $n$ vertices. Assume that $G$ is crossing-augmented. Suppose that there exists a pair of crossing edges in $G$ having a cycle edge $e$ that is crossed. As observed in Section~\ref{sse:background}, a distinct copy $e'$ of $e$ can be added to $G$ so that $e'$ does not cross any edge in $G$. We call $e'$ the \emph{planar copy of $e$}. Note that replacing $e$ with $e'$ changes the embedding of $G$, and thus we do not perform this operation. However, the definition of planar copy will be useful in the following.    

In the next lemma we show how to use a given edge partition of $G$ to compute an embedding-preserving \opvr of $G$ whose vertex complexity is at most twice the maximum vertex degree of the red graph. 

\begin{lemma}\label{le:decomp-opvr}
Let $G$ be an $n$-vertex 3-connected 1-plane graph with a given edge partition such that the maximum vertex degree of the red graph is $d$. There exists an $O(n)$-time algorithm that computes an embedding-preserving \opvr of $G$ with vertex complexity at most $2d$ on an integer grid of size $O(n) \times O(n)$.
\end{lemma}
\begin{proof}
The proof is based on an algorithm that works in three steps. In the first step we augment $G_B$ to a maximal plane graph and we compute a BVR of the resulting graph. In the second step the edges of $G_R$ are inserted in the BVR computed in the first step. In the third step an \opvr of $G$ is computed.

\smallskip\noindent\textbf{Step 1: BVR computation.} We first show how to augment $G_B$ to a maximal plane graph $G^*_B$ such that, for each pair of crossing edges in $G$, the corresponding four cycle edges belong to $G^*_B$. For each cycle edge $e$ of $G$ that is colored red (and thus belongs to $G_R$), we introduce in $G_B$ its planar copy $e'$. Afterwards, we augment the graph (by adding edges) until it is maximal  (which can be done in $O(n)$ time without introducing multiple edges, see, e.g.,~\cite{DBLP:conf/swat/KantB92}). Since $G^*_B$ is maximal, a strong BVR of $G^*_B$ using the $\varepsilon$-visibility model can be computed as follows. We first orient the edges of $G^*_B$ such that the resulting directed graph is acyclic and contains a single source $s$ and a single sink $t$ on its outer face, i.e., it is a plane $st$-graph. This can be done in $O(n)$ time (see, e.g.,~\cite{DBLP:journals/dcg/RosenstiehlT86,TamassiaTollis86}). A strong BVR\footnote{To avoid confusion, it might be worth observing that the terminology used in~\cite{TamassiaTollis86} is slightly different. In particular, a strong BVR using the $\varepsilon$-visibility model (i.e., the model we are referring to in this point) is called an $\varepsilon$-visibility representation in~\cite{TamassiaTollis86}. } $\gamma$ of the plane $st$-graph $G^*_B$ can be computed in $O(n)$ time, such that $s$ and $t$ are represented by the bottommost and the topmost bars of $\gamma$, respectively~\cite{TamassiaTollis86}.
For our purposes we choose as $s$ and $t$ two vertices that belong to the outer face of $G$ (and therefore of $G^*_B$). Observe that, since $G$ is 1-plane, it has at least two vertices on the outer face.  With this choice we can prove the following claim, which will be useful in the remainder of the proof, recall that $G_B^*$ has been constructed such that it contains all cycle edges for each pair of crossing edges in $G$.

\begin{figure}[t]
    \centering
    \includegraphics[scale=0.6,page=3]{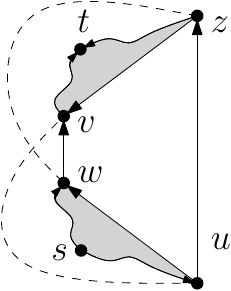}
    \caption{Illustration for the proof of Claim~\ref{cl:bconf}.\label{fi:transitive}}
\end{figure}

\begin{claim}\label{cl:bconf}
For each augmented B-configuration of $G$ formed by two crossing edges $(u,v)$ and $(w,z)$, the four cycle edges of $(u,v)$ and $(w,z)$ are oriented in $G_B^*$ such that one of them is a transitive edge for the cycle  $\{u,w,v,z\}$.
\end{claim}
\begin{proof}
Refer to Figure~\ref{fi:transitive}. Let $(u,v)$ and $(w,z)$ be a pair of crossing edges forming an augmented B-configuration in $G$, such that vertices $w$ and $v$ lie inside the cycle composed of the edge $(u,z)$, the part of the edge $(u,v)$ from $u$ to the crossing point, and the part of the edge $(w,z)$ from the crossing point to $z$. Suppose, for a contradiction, that the cycle $\{u,w,v,z\}$ does not contain a transitive edge in $G^*_B$. In other words, both the left path and the right path of this cycle contain one vertex. Then consider the two ``inner'' vertices of the B-configuration, $v$ and $w$. Since the  cycle contains no transitive edge, either $v$ or $w$, say $v$, is either the origin or the destination of the cycle. Suppose that $v$ is the destination of the cycle (if it is the origin the proof is symmetric). Since $G^*_B$ is 3-connected, there is at least one path from $v$ to $t$ in $G^*_B$. This path can cross neither the edge $(w,z)$ nor the edge $(u,v)$ in $G$ because $(w,z)$ and $(u,v)$ already cross each other. This path cannot cross edge $(u,z)$ as otherwise $G^*_B$ would not be plane (edge $(u,z)$ is a cycle edge and hence either it belongs to $G_B$ and thus to $G^*_B$, or its planar copy has been added to $G^*_B$). It follows that $t$ lies inside the cycle formed by the part of the edge $(u,v)$ from $u$ to the crossing point, the part of the edge $(w,z)$ from the crossing point to $z$, and the edge $(u,z)$. This contradicts the fact that $t$ is on the outer face of $G$.
\end{proof}

\smallskip\noindent\textbf{Step 2: Insertion of the edges of $G_R$.} We now show how to modify $\gamma$ in order to insert the edges of $G_R$. Let $(w,z)$ be an edge of $G_R$, and let $(u,v)$ be the edge of $G^*_B$ crossed by $(w,z)$. Since $G$ is crossing-augmented, if two pairs of crossing edges form a W-configuration, then each of the two pairs forms either a kite or an augmented B-configuration. Similarly, if three pairs of crossing edges form a T-configuration, then each single pair forms either a kite or an augmented B-configuration. Based on this observation, we distinguish whether $(u,v)$ and $(w,z)$ form a kite or an augmented B-configuration in $G$.

\smallskip\noindent{\bf Case A.}  Edges $(u,v)$ and $(w,z)$ form a kite in $G$. Then we further distinguish between the two cases where the cycle $\{u,w,v,z\}$ has a transitive edge in $G^*_B$ or not.

\begin{figure}[t]
    \centering
    \begin{minipage}[b]{.31\textwidth}
    	\centering
    	\includegraphics[scale=0.5,page=2]{reinsertion}
    	\subcaption{}\label{fi:reinsertion-a1-1}
    \end{minipage}
     \begin{minipage}[b]{.31\textwidth}
    	\centering
    	\includegraphics[scale=0.5,page=4]{reinsertion}
    	\subcaption{}\label{fi:reinsertion-a1-2}
    \end{minipage}
    \begin{minipage}[b]{.31\textwidth}
    	\centering
    	\includegraphics[scale=0.5,page=6]{reinsertion}
    	\subcaption{}\label{fi:reinsertion-a1-5}
    \end{minipage}
    \\
    \begin{minipage}[b]{.31\textwidth}
    	\centering
    	\includegraphics[scale=0.5,page=3]{reinsertion}
    	\subcaption{}\label{fi:reinsertion-a1-3}
    \end{minipage}
     \begin{minipage}[b]{.31\textwidth}
    	\centering
    	\includegraphics[scale=0.5,page=1]{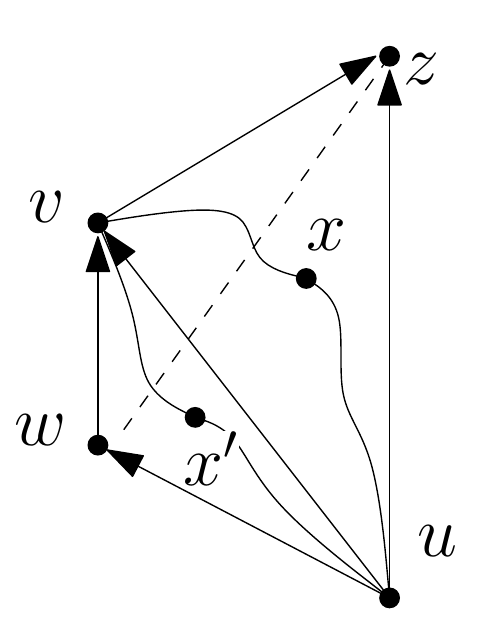}
    	\subcaption{}\label{fi:reinsertion-a1-4}
    \end{minipage}
    \caption{Illustration for {\bf Case A1} of the proof of Theorem~\ref{th:decomp-opvr}.}
\end{figure}

{\bf Case A.1.} Suppose first that the cycle $\{u,w,v,z\}$ has a transitive edge, say $(u,z)$, as shown in Figure~\ref{fi:reinsertion-a1-1}. Let $\sigma$ be the vertical segment that connects $u$ and $z$ and passes through the rightmost point of $v$, as shown in Figure~\ref{fi:reinsertion-a1-2}. We claim that there is no horizontal bar of a vertex that shares an inner point with $\sigma$.  If a vertex $x$ had one such horizontal bar, then it would see both $u$ and $v$ inside the region of $\gamma$ bounded by $(u,v)$, $(v,z)$, and $(u,z)$ (see also Figure~\ref{fi:reinsertion-a1-3}). Since $\gamma$ is a strong BVR, there would exist a path connecting $u$ and $v$ and containing $x$. Such a path would cross the edge $(w,z)$ in $G$, which is impossible because $G$ is 1-planar and $(w,z)$ is already crossed by $(u,v)$ (see also Figure~\ref{fi:reinsertion-a1-4}).

Analogously, let $\sigma'$ be the vertical segment connecting $v$ to the rightmost point of $w$, as shown in Figure~\ref{fi:reinsertion-a1-2}. We claim that there is no horizontal bar that shares an inner point with $\sigma'$. If a vertex $x'$ had one such horizontal bar, then it would see both $u$ and $v$ inside the region of $\gamma$ bounded by $(u,w)$, $(w,v)$, and $(u,v)$ (see also Figure~\ref{fi:reinsertion-a1-3}). Since $\gamma$ is a strong BVR, there would exist a path connecting $u$ and $v$ and containing $x'$. Such a path would cross the edge $(w,z)$ in $G$, which is impossible because $G$ is 1-planar and $(w,z)$ is already crossed by $(u,v)$ (see also Figure~\ref{fi:reinsertion-a1-4}).

It follows that if a horizontal bar intersects $\sigma$ or $\sigma'$, this intersection happens at an endpoint of the bar. Since we are using the $\varepsilon$-visibility model, the bar can be slightly shortened so not to intersect $\sigma$ or $\sigma'$ anymore. Then we can use $\sigma$ and $\sigma'$  to draw one vertical bar attached to the horizontal bar of $w$ and one vertical bar attached to the horizontal bar of $z$, such that they are contained in $\sigma$ and $\sigma'$, respectively, and they see each other through a horizontal visibility that crosses $(u,v)$, see Figure~\ref{fi:reinsertion-a1-5}.

{\bf Case A.2.} Suppose now that the cycle $\{u,w,v,z\}$ has no transitive edge, as shown in Figure~\ref{fi:reinsertion-a2-1}, and hence is drawn as in Figure~\ref{fi:reinsertion-a2-2}. By applying a similar argument as above, we can draw two vertical bars attached to the horizontal bars of $w$ and $z$, respectively, and such that they see each other crossing $(u,v)$, as shown in Figure~\ref{fi:reinsertion-a2-3}.

\begin{figure}[t]
    \centering  
    \begin{minipage}[b]{.31\textwidth}
    	\centering
    	\includegraphics[scale=0.5,page=5]{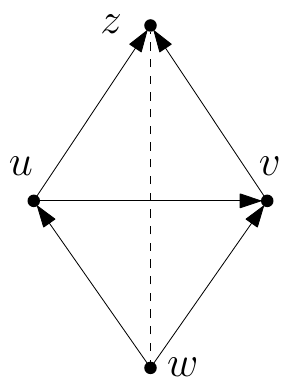}
    	\subcaption{}\label{fi:reinsertion-a2-1}
    \end{minipage}
    \begin{minipage}[b]{.31\textwidth}
    	\centering
    	\includegraphics[scale=0.5,page=7]{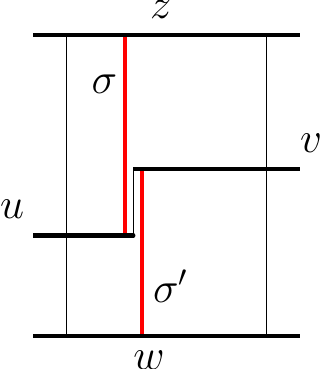}
    	\subcaption{}\label{fi:reinsertion-a2-2}
    \end{minipage}
    \begin{minipage}[b]{.31\textwidth}
    	\centering
    	\includegraphics[scale=0.5,page=8]{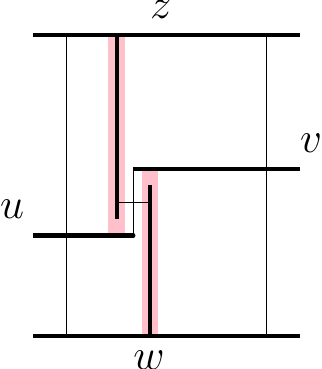}
    	\subcaption{}\label{fi:reinsertion-a2-3}
    \end{minipage}
    \begin{minipage}[b]{.31\textwidth}
    	\centering
    	\includegraphics[scale=0.5,page=9]{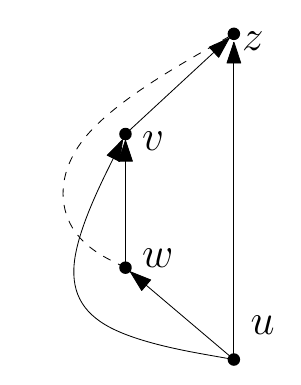}
    	\subcaption{}\label{fi:reinsertion-b-1}
    \end{minipage}
    \begin{minipage}[b]{.31\textwidth}
    	\centering
    	\includegraphics[scale=0.5,page=10]{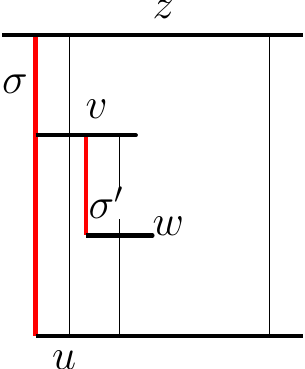}
    	\subcaption{}\label{fi:reinsertion-b-2}
    \end{minipage}
    \begin{minipage}[b]{.31\textwidth}
    	\centering
    	\includegraphics[scale=0.5,page=11]{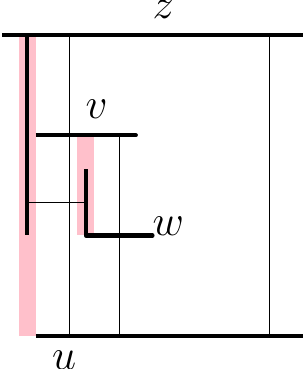}
    	\subcaption{}\label{fi:reinsertion-b-3}
    \end{minipage}
    \caption{Illustration for (a)--(c) {\bf Case A.2} and (d)--(f) {\bf Case B} of the proof of Theorem~\ref{th:decomp-opvr}.}
\end{figure}

\smallskip\noindent{\bf Case B.}  Edges $(u,v)$ and $(w,z)$ form an augmented B-configuration in $G$. By Claim~\ref{cl:bconf}, the cycle $\{u,w,v,z\}$ always has a transitive edge in $G^*_B$, as shown in Figure~\ref{fi:reinsertion-b-1}, and hence is drawn as in Figure~\ref{fi:reinsertion-b-2}. Then it can be handled analogously as in the above cases, as shown in Figure~\ref{fi:reinsertion-b-3}.

Once all red edges have been inserted, we remove the vertical visibilities representing the planar copies of red cycle edges introduced in \textbf{Step 1}, if any. We conclude the description of this step by observing that it can be performed in $O(n)$ time, since each red edge can be reinserted in $O(1)$ time. 

\smallskip\noindent\textbf{Step 3: Computation of the \opvr of $G$.} Denote by $\gamma^*$ the visibility representation of $G$ obtained after {\bf Step 2}. Each edge of $G$ is now represented as vertical or horizontal visibility between two horizontal or two vertical bars, respectively. Since for each vertex we inserted at most $d$  edges, each vertex is represented by a ``rake''-shaped object with one horizontal bar and at most $d$ vertical bars. By slightly thickening these objects, we obtain orthogonal polygons with at most $2d$ reflex angles, and thus an embedding-preserving \opvr $\gamma$ of $G$ with vertex complexity at most $2d$. In order to obtain an \opvr on an integer grid, we can extract an orthogonal representation $H$ from $\gamma$ and then use the compaction step of the TSM approach. Since $G$ has $O(n)$ crossings and $H$ has $O(n)$ bends, the compaction can be executed in $O(n)$ time and the size of the resulting \opvr is $O(n) \times O(n)$.
\end{proof}

Combining Theorem~\ref{th:edge-decomposition} and Lemma~\ref{le:decomp-opvr}, we obtain the following theorem.

\begin{theorem}\label{th:decomp-opvr}
Let $G$ be a 3-connected 1-plane graph with $n$ vertices. There exists an $O(n)$-time algorithm that computes an embedding-preserving \opvr of $G$ with vertex complexity at most 12 on an integer grid of size $O(n) \times O(n)$.
\end{theorem}
\begin{proof}
By Theorem~\ref{th:edge-decomposition}, every $n$-vertex 3-connected 1-plane graph has an edge partition such that the red graph has maximum vertex degree six, which can be computed in $O(n)$ time. We can exploit this edge partition and apply the algorithm of Lemma~\ref{le:decomp-opvr} to compute an embedding-preserving \opvr of $G$ with vertex complexity at most $12$ on an integer grid of size $O(n)\times O(n)$ in $O(n)$ time.
\end{proof}

Based on Theorem~\ref{th:decomp-opvr}, we can significantly improve the time complexity to compute an optimal \opvr of 3-connected 1-plane graphs.

\begin{theorem}\label{th:3conn-ub}
Let $G$ be a 3-connected 1-plane graph with $n$ vertices. There exists an  $O(n^{\frac{7}{4}}\sqrt{\log n})$-time algorithm that computes an embedding-preserving optimal \opvr $\gamma$ of $G$, on an integer grid of size $O(n) \times O(n)$. Also, $\gamma$ has the minimum number of reflex corners among all the embedding-preserving optimal {\opvr}s of $G$.
\end{theorem}
\begin{proof}
We use the same terminology and notation as in the proof of Lemma~\ref{le:opt}. Let $\gamma$ be an optimal \opvr of $G$ and let $H$ be the corresponding orthogonal representation. By Theorem~\ref{th:decomp-opvr}, $\gamma$ has vertex complexity at most $12$. This implies that $H$ can be computed by executing at most $12$ tests for the existence of a feasible flow for the network $N''$. Also, since the vertex complexity is at most $12$, $H$ has $O(n)$ bends and thus the cost of the flow is $\chi = O(n)$. It follows that $H$ can be computed in time $O(\chi^{\frac{3}{4}}n\sqrt{\log n})=O(n^{\frac{7}{4}}\sqrt{\log n})$.

A drawing of $H$ is obtained by applying the compaction step of the TSM framework. Since the number of bends of $H$ is $O(n)$ and since the number of crossings of a 1-plane graph is $O(n)$ (see, e.g.,~\cite{Suzuki2010}), this step is executed in $O(n)$ time and produces a drawing on an integer grid of  size $O(n) \times O(n)$.
\end{proof}

It is known that there are 3-connected 1-plane graphs such that any  embedding-preserving \opvr has vertex complexity at least one~\cite{SoCG}. To improve this lower bound we use the same graph family as the one used for the tightness of the vertex degree bound in the proof of Theorem~\ref{th:edge-decomposition}.

\begin{theorem}\label{th:3conn-lb}
There exists an infinite family $\mathcal G$ of 3-connected 1-plane graphs such that for any graph $G$ of $\mathcal G$, any embedding-preserving \opvr of $G$ has vertex complexity at least two.
\end{theorem}
\begin{proof}

\begin{figure}[t]
    \centering
    \includegraphics[scale=0.4,page=2]{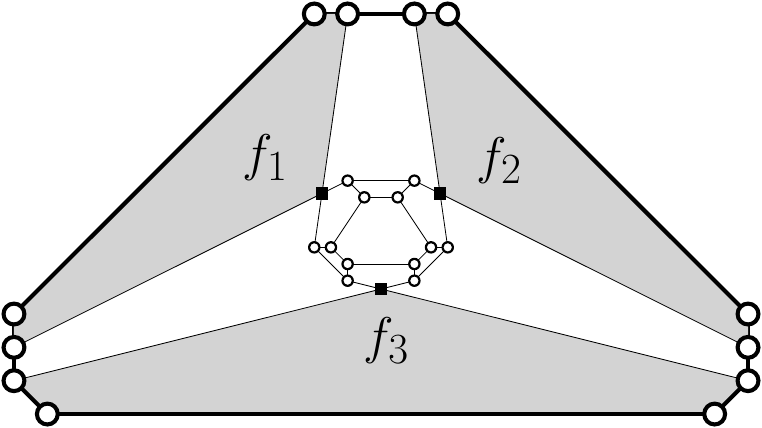}
    \caption{ Illustration for the proof of Theorem~\ref{th:3conn-lb}.}\label{fi:3connlb-2}
\end{figure}

Consider the same family of graphs used to prove the tightness of Theorem~\ref{th:edge-decomposition}, and recall that any $n$-vertex graph $G$ in this family is 3-connected and 1-plane by construction. Since $G$ is 1-plane it admits an \opvr by Theorem~\ref{th:1-planar}.

Consider an \opvr of $G$ and the corresponding orthogonal drawing $\Gamma$. For each face $f$ of $G_p$, there are three faces $f_i$ ($i=1,2,3$) in $G$, such that each of these three faces contains four expansion vertices and one dummy vertex. In Figure~\ref{fi:3connlb-2} the three faces $f_i$ ($i=1,2,3$) for the face $f$ in Figure~\ref{fi:3connlb-1} are highlighted. For each face $f_i$, the dummy vertex forms a $\ph$ angle inside $f_i$. Also, each expansion vertex forms one $\ph$ angle. In total there are exactly five $\ph$ vertex-angles inside $f_i$.  Then, since the real edges of $f_i$ do not have bends, by Property~\ref{pr:face} one of the two expansion edges must form (at least) one bend-angle of $\pt$ inside $f_i$, and therefore a bend-angle of $\ph$ inside the corresponding expansion cycle. Since there are $2n-4$ faces in $G_p$, there are $6n-12$ faces of $\Gamma$ each requiring at least one $\pt$ angle from an expansion edge. If every vertex of $G$ is represented by a polygon with vertex complexity at most $k$, then the edges of each expansion cycle form at most $4+k$ angles of $\pt$ inside their incident faces (that are not expansion cycles). Hence it must be that $(4+k)n \geq 6n-12$, that is $k \geq 2 - \frac{12}{n}$. Since $k$ is an integer, it follows that $k \geq 2$ for any $n \geq 13$.
\end{proof}

\section{Vertex Complexity Bounds for 2-connected 1-plane Graphs}\label{se:2conn-bounds}

In this section, we show that if an $n$-vertex 1-plane graph $G$ is 2-connected and it can be augmented to become 3-connected only at the expense of losing its 1-planarity, then the vertex complexity of any \opvr of $G$ may be $\Omega(n)$. Also, for these graphs we show that a 1-planar embedding that guarantees constant vertex complexity can be computed in $O(n)$ time under the assumption that they do not have a certain type of crossing configuration.

\begin{theorem}\label{th:2conn-lb}
For every positive integer $n$, there exists a 2-connected 1-planar graph $G$ with $O(n)$ vertices such that, for every 1-planar embedding of $G$, any embedding-preserving \opvr of $G$ has vertex complexity $\Omega(n)$.
\end{theorem}
\begin{proof}
We first prove the claim for a fixed 1-planar embedding. Consider the 1-plane graph $K$ in Figure~\ref{fi:2connlb-1}. It has 2 vertices on its outer face, $u$ and $v$, plus 6 inner vertices. We now construct a graph $G$ as follows. Attach $n+1$ copies $K_1, \dots, K_{n+1}$ of $K$ such that they all share $u$ and $v$. The copies are attached in parallel without introducing any further edge crossing, as shown in Figure~\ref{fi:2connlb-2}. Also connect $u$ and $v$ with an edge on the outer face.  The resulting graph $G$ has $8(n+1)-2n=6n+8=O(n)$ vertices. Also, $G$ is 2-connected and 1-plane by construction. Since it is 1-plane, it admits an \opvr by Theorem~\ref{th:1-planar}. Consider now an embedding-preserving \opvr of $G$ and the corresponding orthogonal drawing $\Gamma$. Between any two consecutive copies $K_i$ and $K_{i+1}$ ($i=1,\dots,n$), there is a face $f_i$ of $G$ having two expansion vertices of $C(u)$ (the expansion cycle of $u$) and two expansion vertices of $C(v)$ on its boundary, together with two dummy vertices; see  Figure~\ref{fi:2connlb-3}. Each dummy vertex forms one $\ph$ angle inside $f_i$. Also, each expansion vertex forms one $\ph$ angle inside $f_i$. Hence, there are at least six $\ph$ angles inside $f_i$. Also, since the real edges of $f_i$ have no bends, by Property~\ref{pr:face} the two expansion edges of $f_i$ must form (at least) two $\pt$ angles inside $f_i$. In $\Gamma$ there are $n$ such faces requiring two angles of $\pt$ each from an expansion edge. If every vertex of $G$ is represented by a polygon with vertex complexity at most $k$, then the edges of each expansion cycle form at most $4+k$ angles of $\pt$ inside their incident faces (that are not expansion cycles). At least nine of these angles are inside the outer face of $\Gamma$ (by Property~\ref{pr:face}), and hence it must be that $(4+k)2-9 \geq 2n$, that is $k > n$. 

\begin{figure}[t]
    \centering
    \begin{minipage}[b]{.31\textwidth}
    	\centering
    	\includegraphics[scale=1,page=1]{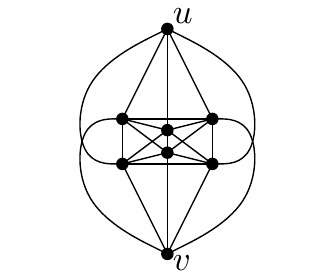}
    	\subcaption{$K$}\label{fi:2connlb-1}
    \end{minipage}
    \begin{minipage}[b]{.31\textwidth}
    	\centering
    	\includegraphics[scale=1,page=2]{2conn1plane-lowerbound}
    	\subcaption{$G$}\label{fi:2connlb-2}
    \end{minipage}
    \begin{minipage}[b]{.31\textwidth}
    	\centering
    	\includegraphics[scale=1,page=3]{2conn1plane-lowerbound}
    	\subcaption{$f_i$}\label{fi:2connlb-3}
    \end{minipage}
    \caption{ Illustration for the proof of Theorem~\ref{th:2conn-lb}.}
\end{figure}

It remains to extend the argument of the proof to any 1-planar embedding of $G$. To this aim, we observe that graph $K$ together with the edge $(u,v)$  has a unique 1-planar embedding~\cite{Suzuki2010}. This, together with the fact that no two copies of $K$ can intersect one another without violating 1-planarity, implies that $G$ has a unique embedding up to a renaming of the $n+1$ copies of $K$, except for the edge $(u,v)$. Such an edge can be indeed placed between any two consecutive copies of $K$. Nonetheless, this does not change the argument above, as there will be a face $f_i$ split in two faces, each requiring at least one $\ph$ angle from either $C(u)$ or $C(v)$.
\end{proof}

The graphs used to prove Theorem~\ref{th:2conn-lb} contain several W-configurations. By contrast, the next theorem shows that the absence of W-configurations suffices to find a 1-planar embedding which admits an \opvr with constant vertex complexity. 

\begin{theorem}\label{th:2conn1plane-ub}
Let $G$ be a 2-connected 1-plane graph with $n$ vertices and no W-confi-gurations. There exists an $O(n)$-time algorithm that computes a 1-planar \opvr of $G$ with vertex complexity at most 22 on an integer grid of size $O(n) \times O(n)$.
\end{theorem}

The proof of Theorem~\ref{th:2conn1plane-ub} is based on a drawing algorithm described in the following. We first construct an $SPQR$-tree~\cite{DBLP:journals/siamcomp/BattistaT96} of the input graph $G$ (see also Appendix~\ref{ap:spqrtree}), and then traverse  the tree bottom-up in order to compute the desired representation, while maintaining a set of geometric invariants. $R$-nodes, which are the most challenging components, are handled with a sophisticated variant of the technique used to prove Theorem~\ref{th:decomp-opvr}. The embedding of the computed \opvr may be different from the one of $G$, but it is still 1-planar, i.e., it has at most one crossing per edge. More precisely, we may need to use flip and swap operations on the $SPQR$-tree of $G$.  Analogously to  the proof of Theorem~\ref{th:decomp-opvr}, we describe how to construct a visibility representation where vertices of $G$ are connected geometric features composed of horizontal and vertical bars (possibly not ``rake''-shaped in this case, but still arranged in a tree-like structure). These objects are then used as ``skeletons'' for the orthogonal polygons.

Let $T$ be the $SPQR$-tree of $G$ rooted at a $Q$-node $\rho$. We assume that $G$ is crossing-augmented. This implies the following property:

\begin{property}\label{pr:crossings-in-rnodes}
Let $(u,v)$ and $(w,z)$ be two edges that cross each other in $G$. Then the two corresponding $Q$-nodes are both children of the same $R$-node.
\end{property}

As a consequence of Property~\ref{pr:crossings-in-rnodes}, we also have:

\begin{property}\label{pr:virtualedges}
Let $\mu$ be an $R$-node, and let $e_1$ and $e_2$ be two virtual edges of its skeleton $sk(\mu)$ such that they cross each other. Then $e_1$ and $e_2$ correspond to two $Q$-nodes.
\end{property}

For the sake of description, we also apply the following transformation to $G$ and to its $SPQR$-tree $T$. Let $\mu$ be a $P$-node of $T$ with a $Q$-node $\nu$ as a child. If the parent of $\mu$ is not an $R$-node we subdivide the edge $e$ corresponding to $\nu$ with a subdivision vertex. This corresponds to replacing $\nu$ with an $S$-node having two $Q$-nodes as children. Also, note that since the parent of $\mu$ is not an $R$-node, it can only be an $S$-node (this observation will be useful in the following).

Denote by $G'$ the graph obtained from $G$ by applying the above operation to all $P$-nodes of $T$. Let $T'$ be the resulting $SPQR$-tree.  The algorithm performs a bottom-up visit of $T'$ and computes for each visited node $\mu$ a visibility representation of the pertinent graph $G_{\mu}$ of $\mu$. The leaves of $T'$ (i.e.,  the $Q$-nodes) are ignored, since the corresponding edges are drawn as visibilities in the visibility representation of their parent nodes. Let $\mu$ be a node of $T'$ different from the root $\rho$, and let $G_\mu$ be the corresponding pertinent graph whose poles are $s_\mu$ and $t_\mu$. The algorithm computes a visibility representation $\gamma_\mu$ of $G_\mu$ that satisfies the following invariants.

\begin{description}
\item[\bf I1.] Vertex $s_\mu$ is represented by one horizontal bar that is the bottommost bar of $\gamma_\mu$.

\item[\bf I2.] Vertex $t_\mu$ is represented by one vertical bar that is the leftmost bar of $\gamma_\mu$.

\item[\bf I3.] Every vertex different from $s_\mu$ and $t_\mu$ is represented by a set of at most 12 bars.
\end{description}

We now show how to construct $\gamma_\mu$ based on whether $\mu$ is an $R$-node, a $P$-node, or an $S$-node. The root $\rho$ and its child $\xi$ are handled in a special way.

\begin{figure}[t]
    \centering
    \begin{minipage}[b]{.31\textwidth}
    	\centering
    	\includegraphics[scale=0.3,page=1]{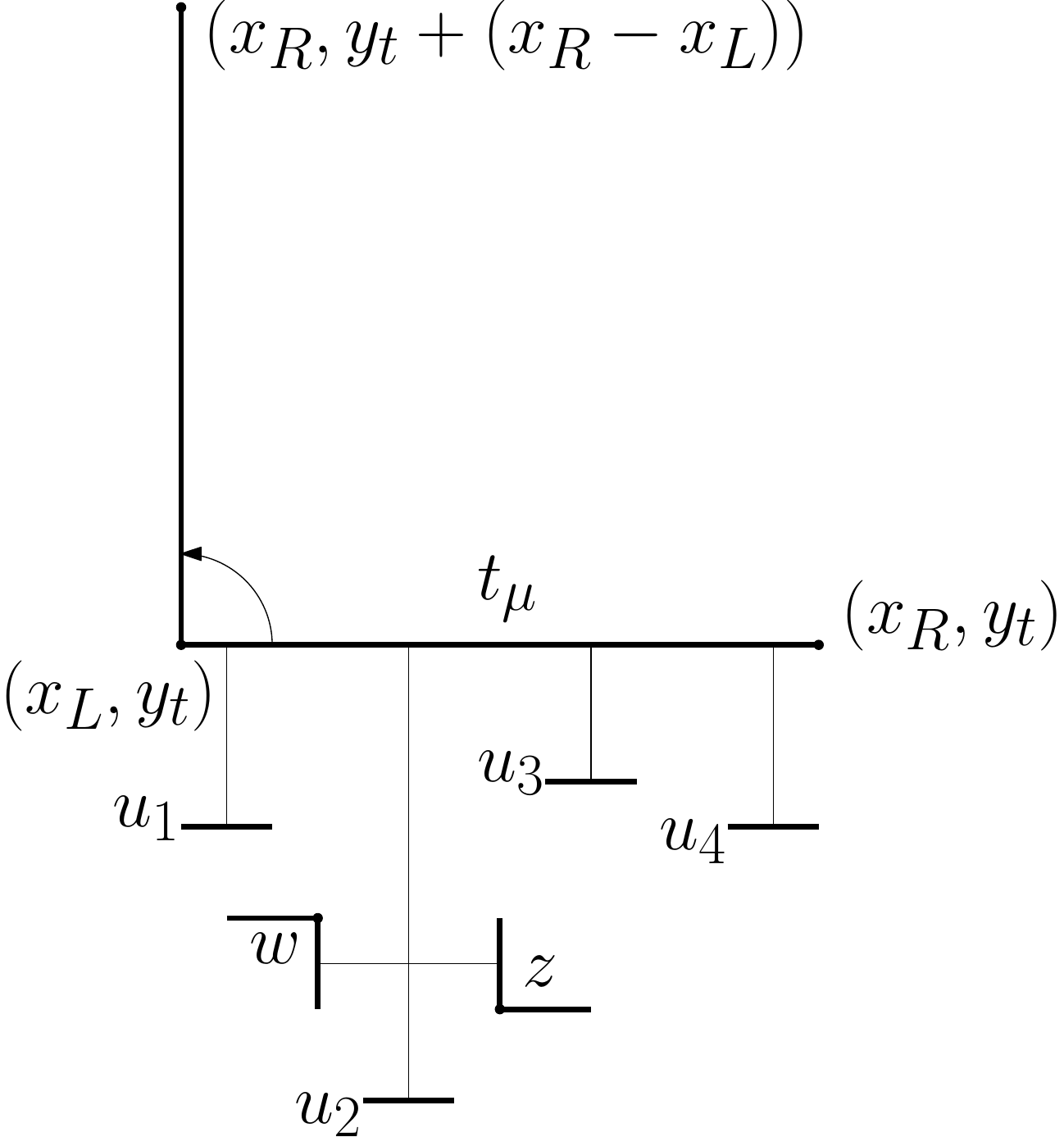}
    	\subcaption{}\label{fi:rnodes-1}
    \end{minipage}
    \begin{minipage}[b]{.31\textwidth}
    	\centering
    	\includegraphics[scale=0.3,page=2]{rnodes}
    	\subcaption{}\label{fi:rnodes-2}
    \end{minipage}
    \begin{minipage}[b]{.31\textwidth}
    	\centering
    	\includegraphics[scale=0.3,page=3]{rnodes}
    	\subcaption{}\label{fi:rnodes-3}
    \end{minipage}
    \begin{minipage}[b]{.31\textwidth}
    	\centering
    	\includegraphics[scale=0.3,page=4]{rnodes}
    	\subcaption{}\label{fi:rnodes-4}
    \end{minipage}
    \begin{minipage}[b]{.31\textwidth}
    	\centering
    	\includegraphics[scale=0.3,page=5]{rnodes}
    	\subcaption{}\label{fi:rnodes-5}
    \end{minipage}
    \begin{minipage}[b]{.31\textwidth}
    	\centering
    	\includegraphics[scale=0.3,page=6]{rnodes}
    	\subcaption{}\label{fi:rnodes-6}
    \end{minipage}
    \caption{ Illustration for the proof of Lemma~\ref{le:rnode}.}
\end{figure}

\begin{lemma}\label{le:rnode}
Let $\mu$ be an $R$-node of $T'$. Let $\gamma_{\nu_1}$, $\gamma_{\nu_2}$, $\dots$, $\gamma_{\nu_k}$ be the $k \geq 0$ visibility representations of the $k \geq 0$ children of $\mu$ (excluding $Q$-nodes) for which Invariants {\bf I1--I3} hold. Then $G_\mu$ admits a visibility representation $\gamma_\mu$ that respects Invariants {\bf I1--I3}.
\end{lemma}
\begin{proof}
The idea is to first compute a visibility representation of the skeleton $sk(\mu)$ of $G_\mu$ (ignoring the reference edge $(s_\mu,t_\mu)$), and then replace each virtual edge $e_{\nu_k}$ associated with node ${\nu_k}$ with the corresponding visibility representation $\gamma_{\nu_k}$. In particular, recall that $sk(\mu)$ is a 3-connected 1-plane graph (see Appendix~\ref{ap:spqrtree}), and hence we can compute an edge partition of $sk(\mu)$ such that the red graph has maximum vertex degree $6$ (Theorem~\ref{th:edge-decomposition}). Denote by $sk_R(\mu)$ and by $sk_B(\mu)$ the red and blue graph, respectively. Applying {\bf Step 1} of the proof of Theorem~\ref{th:decomp-opvr}, we obtain a strong BVR of $sk_B(\mu)$.  Recall that {\bf Step 1} of Theorem~\ref{th:decomp-opvr} requires the choice of two vertices, $s$ and $t$ on the outer face of $sk(\mu)$, such that $sk_B(\mu)$ will have $s$ and $t$ as unique source and sink, respectively. In our case, we choose $s=s_\mu$ and $t=t_\mu$. Afterwards, by applying {\bf Step 2} of the proof of Theorem~\ref{th:decomp-opvr}, we obtain a visibility representation $\gamma^*$ of  $sk(\mu)$ such that each vertex is represented by a horizontal bar plus at most six vertical bars. Since $G$ contains no W-configurations, there is at most one pair of edges  $(s_\mu,u)$ and $(t_\mu,v)$ that cross at a point $p$ that belongs to the outer face of $sk(\mu) \setminus \{(s_\mu,t_\mu)\}$. If such a pair does not exist, then both $s_\mu$ and $t_\mu$ are represented by a single horizontal bar. Else, according to the technique used in {\bf Step 2} of Theorem~\ref{th:decomp-opvr},  we can choose one vertex between $s_\mu$ and $t_\mu$ to be represented by a single horizontal bar, while the other one is represented by a horizontal bar plus a vertical bar. We choose $t_\mu$ to be represented by a single horizontal bar.

Let $\Pi_1$ and $\Pi_2$ be the two paths between $s_\mu$ and $t_\mu$ that belong to the outer face of $sk(\mu) \setminus \{(s_\mu,t_\mu)\}$ (which is a cycle since $sk(\mu)$ is 3-connected). As observed above, at most one of them can be composed of exactly one crossing point $p$ connected to both $s_\mu$ and $t_\mu$.  Without loss of generality, we can assume that such a path (if any) is $\Pi_1$, as otherwise we can flip $G_\mu$ around $s_\mu$ and $t_\mu$ such that this is the case. We now show how to transform $\gamma^*$ into a visibility representation $\gamma'$ that satisfies Invariants {\bf I1--I3}. Let $(x_L,y_t)$ and $(x_R,y_t)$ be the leftmost and rightmost points of the horizontal bar representing $t_\mu$ in $\gamma^*$. We transform this horizontal bar into the vertical bar having $(x_L,y_t)$ and $(x_L,y_t+(x_R-x_L))$  as bottommost and topmost points, respectively; see also Figure~\ref{fi:rnodes-1} for an illustration. In other words, we rotate the bar by $\ph$ counter-clockwise around point $(x_L,y_t)$. Moreover, if $x_L$ is not the leftmost point of $\gamma^*$, we can translate the vertical bar further to the left, until its $x$-coordinate is the leftmost one. This will ensure {\bf I2}. We now describe how to transform all the blue edges incident to $t_\mu$ (drawn as vertical visibilities in $\gamma^*$) into horizontal visibilities in $\gamma'$, and how to transform all the red edges (drawn as horizontal visibilities in $\gamma^*$) crossing the blue edges incident to $t_\mu$ into vertical visibilities in $\gamma'$. Let $(u_1,t_\mu)$, $\dots$, $(u_h,t_\mu)$, be the $h > 0$ blue edges incident to $t_\mu$ ordered according to the left-to-right total order defined by the corresponding vertical visibilities in $\gamma^*$. In other words, for any two blue edges $(u_i,t_\mu)$ and $(u_j,t_\mu)$, $i < j$ if and only if the vertical visibility representing  $(u_i,t_\mu)$ in $\gamma^*$ is to the left of the vertical visibility representing $(u_j,t_\mu)$ in $\gamma^*$.  Let $(x_i,y_i)$ be the bottommost point of the vertical visibility representing $(u_i,t_\mu)$ in $\gamma^*$. We replace this visibility with a vertical bar in $\gamma'$ having $(x_i,y_i)$ and $(x_i,y_t+i)$ as bottommost and topmost points, respectively; see also Figure~\ref{fi:rnodes-2} for an illustration. This adds one vertical bar to the representation of $u_i$. After applying this operation for all edges  $(u_i,t_\mu)$, $i=1,\dots,h$, we have that every vertex $u_i$ can see $t_\mu$ through a horizontal visibility. Also, the circular order of the edges around $t_\mu$ is preserved from $\gamma^*$ to $\gamma'$. We remark that, at this point, every object representing a vertex consists of at most 8 connected bars.

For each red edge $(w,z)$, consider the two vertical bars in $\gamma^*$ such that $(w,z)$ is a horizontal visibility between them, and let $(x_w,y_w)$ and $(x_z,y_z)$ be the coordinates of the bottommost point of the vertical bar of $w$ and $z$, respectively. We now show how to replace each red edge $(w,z)$ crossing a blue edge $(u_i,t_\mu)$ with a vertical visibility crossing the horizontal visibility that represents $(u_i,t_\mu)$ in $\gamma'$. Without loss of generality, we can assume that $x_w < x_z$. Then we extend the vertical bar of $z$ such that its topmost point is now $(x_z,y_t+i+\frac{1}{4})$, and then attach a horizontal bar whose rightmost point is $(x_z,y_t+i+\frac{1}{4})$ and its leftmost point is $(x_w-\varepsilon,y_t+i+\frac{1}{4})$ (for an arbitrarily small value of $\varepsilon$ in order to ensure the $\varepsilon$-visibility model). Also, we remove the vertical bar of $w$. With this construction  $w$ and $z$ see each other through a vertical visibility that crosses $(u_i,t_\mu)$; see also Figure~\ref{fi:rnodes-3} for an illustration. Notice that if the pair of crossing edges  $(s_\mu,u)$ and $(t_\mu,v)$ exists, then $(s_\mu,u)$ is red and $x_{s_\mu} < x_u$. Hence this transformation removes the vertical bar of $s_\mu$, which is now represented by a single horizontal bar and therefore {\bf I1} holds.

\begin{figure}[t]
    \centering
    \includegraphics[scale=0.7]{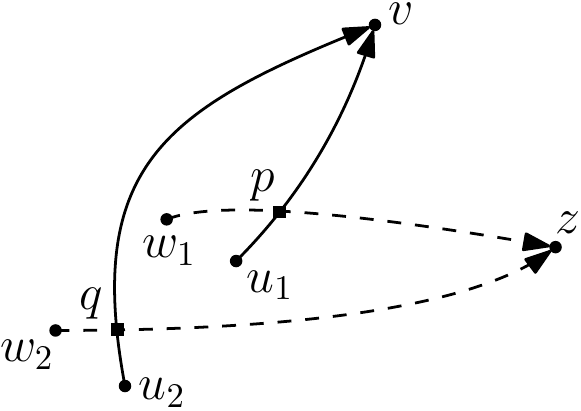}
    \caption{Illustration for the proof of Lemma~\ref{le:rnode}; red (blue) edges are dashed (solid).\label{fi:rnode-orientation}}
\end{figure}

Observe that for each red edge $(w,z)$ to which the above transformation is applied, an additional horizontal bar is added to its right end-vertex ($z$ in our description). We show that, for each vertex $z$, we have at most one such edge, and therefore at most one additional bar.  As a consequence each vertex  is represented by a set of bars of size at most $9$. More precisely, we claim that each vertex $z$ is incident to at most one red edge $(w,z)$ that crosses a blue edge incident to $t_\mu$, and such that $x_w < x_z$. To prove this, we orient each red edge $(w,z)$ from $w$ to $z$ if $x_w < x_z$, and each blue edge $(u,v)$ from $u$ to $v$ if the horizontal bar of $u$ is below the horizontal bar of $v$. Then the red edges always cross the blue edges from left to right. Equivalent to our claim, we show that for each vertex $z$ there are no two incoming red edges that cross two blue edges incident to the same vertex. Suppose, for a contradiction, that two red edges $(w_1,z)$ and $(w_2,z)$ cross two blue edges $(u_1,v)$ and $(u_2,v)$ at points $p$ and $q$; see also Figure~\ref{fi:rnode-orientation} for an illustration.  Then due to the orientation of these four edges and due to 1-planarity, at least two vertices are inside the cycle $\{z,p,v,q\}$, and at least two vertices are outside this cycle. Since every path connecting a vertex inside the cycle to a vertex outside the cycle passes through $v$ or $z$, $v$ and $z$ are a separation pair of $sk(\mu)$ --  a contradiction since $sk(\mu)$ is 3-connected.

It remains to replace each virtual edge $e_{\nu_i}=(u,v)$ in $\gamma'$ with the corresponding visibility representation $\gamma_{\nu_i}$; see also Figure~\ref{fi:rnodes-4}. Note first that if the edge $(u,v)$ exists in the pertinent graph $G_{\nu_i}$ of $\nu_i$, then it is already drawn in $\gamma'$ as the visibility representing $e_{\nu_i}$ (although using such a drawing may imply a swap operation between $(u,v)$ and the rest of $G_{\nu_i}$). Since {\bf I1-I2} hold for $\gamma_{\nu_i}$, we can assume that $u$ is the leftmost vertical bar of $\gamma_{\nu_i}$, and $v$ is the bottommost horizontal bar of $\gamma_{\nu_i}$. The idea is to merge $\gamma_{\nu_i}$ in $\gamma'$, possibly scaling and/or stretching\footnote{By Invariants {\bf I1} and {\bf I2}, $u$ has only horizontal visibilities and $v$ has only vertical visibilities, and hence we can translate $u$ horizontally and $v$ vertically and then scale $\gamma_{\nu_i}$ so to fit it in a prescribed region of the plane.} $\gamma_{\nu_i}$. This operation increases  the number of bars of either $u$ or $v$ by one unit. If we merge $\gamma_{\nu_i}$ in $\gamma'$ without rotating it, then the number of bars of $v$ is increased by one; see also Figure~\ref{fi:rnodes-5} for an illustration. Else, if we rotate $\gamma_{\nu_i}$ clockwise by $\ph$, then $u$ is the vertex receiving one more bar; see also Figure~\ref{fi:rnodes-6} for an illustration.

With an idea similar to the one used in the proof of Theorem~\ref{th:edge-decomposition}, we can exploit Schnyder trees to obtain a 3-orientation of $sk_B(\mu)$, and, for each virtual edge $e_{\nu_i}=(u,v)$ oriented from $u$ to $v$, ``charge'' the additional bar on $u$. This ensures that each vertex is charged with at most 3 further bars, which leads to at most $12$ bars per vertex, and thus {\bf I3} holds. Note that edge $(s_\mu,t_\mu)$ is the reference edge of $\mu$ and hence it is not replaced with a visibility representation, regardless of its orientation. Also, in a 3-orientation obtained via Schnyder trees, all other edges of $sk_B(\mu)$ incident to $s_\mu$ and $t_\mu$ can be oriented incoming with respect to both $s_\mu$ and $t_\mu$, which guarantees that a bar is added to neither $s_\mu$ nor $t_\mu$ and thus {\bf I1-I2} hold. More precisely, all edges of $sk_B(\mu)$ that are not part of the outer face are oriented incoming with respect to both $s_\mu$ and $t_\mu$ by construction. Furthermore, the two edges of the outer face distinct from $(s_\mu,t_\mu)$ and incident one to $s_\mu$ and one to $t_\mu$ can be freely oriented, and therefore they can be oriented incoming with respect to $s_\mu$ and $t_\mu$.
\end{proof}

\begin{figure}[t]
\centering
\includegraphics[scale=0.4,page=1]{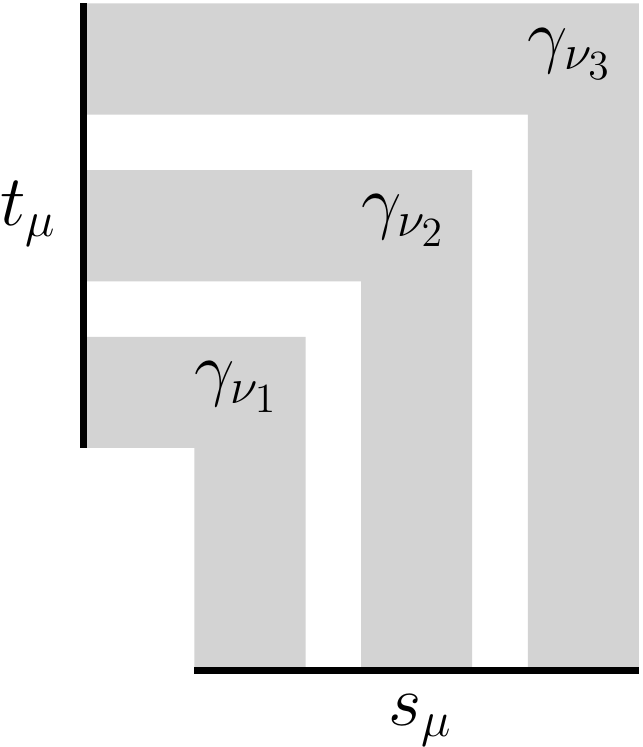}
\caption{ Illustration for the proof of Lemma~\ref{le:pnodes}.}\label{fi:pnodes}
\end{figure}

\begin{lemma}\label{le:pnodes}
Let $\mu$ be a $P$-node of $T'$. Let $\gamma_{\nu_1}$, $\gamma_{\nu_2}$, $\dots$, $\gamma_{\nu_k}$ be the $k > 0$ visibility representations of the $k \geq 0$ children of $\mu$ (excluding $Q$-nodes) for which Invariants {\bf I1--I3} hold. Then $G_\mu$ admits a visibility representation $\gamma_\mu$ that respects Invariants {\bf I1--I3}.
\end{lemma}
\begin{proof}
If $\mu$ is not the child of an $R$-node, then it does not have any $Q$-nodes among its children (if it had one in $G$, then it was subdivided). Else, if the parent of $\mu$ is an $R$-node $\nu$ and $\mu$ has a $Q$-node as a child, then, as explained in the proof of Lemma~\ref{le:rnode}, the edge corresponding to this $Q$-node is drawn directly in the visibility representation $\gamma_\nu$. Thus, we can assume that $\mu$ does not have any $Q$-node among its children.

Since the visibility representations $\gamma_{\nu_1}$, $\gamma_{\nu_2}$, $\dots$, $\gamma_{\nu_k}$ satisfy Invariants {\bf I1--I3}, we can suitably scale them and extend the bar representing $s_\mu =s_{\nu_1}=\dots=s_{\nu_k}$ and the bar representing $t_\mu =t_{\nu_1}=\dots=t_{\nu_k}$ so to merge all the drawings as shown in Figure~\ref{fi:pnodes}. This construction satisfies Invariants {\bf I1--I3}.
\end{proof}

\begin{figure}[t]
    \centering
    \begin{minipage}[b]{.32\textwidth}
    	\centering
    	\includegraphics[scale=0.5,page=1]{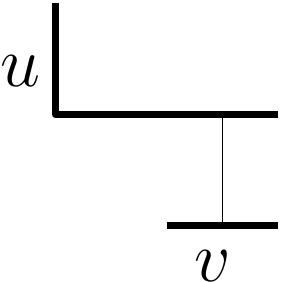}
    	\subcaption{}\label{fi:snodes-1}
    	\vspace{1cm}
    	\includegraphics[scale=0.5,page=2]{snodes}
    	\subcaption{}\label{fi:snodes-2}
    \end{minipage}
    \hfill
    \begin{minipage}[b]{.32\textwidth}
    	\centering
    	\includegraphics[scale=0.5,page=3]{snodes}
    	\subcaption{}\label{fi:snodes-3}
    \end{minipage}
    \begin{minipage}[b]{.32\textwidth}
    	\centering
    	\includegraphics[scale=0.5,page=1]{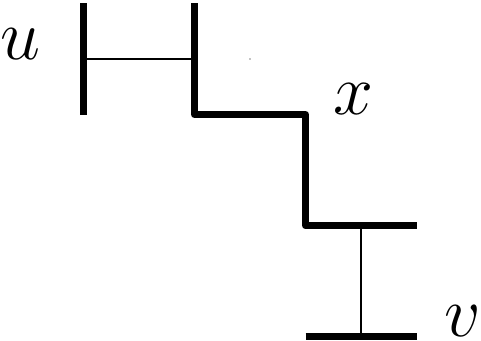}
    	\subcaption{}\label{fi:subdivision-1}
    	\vspace{1cm}
    	\includegraphics[scale=0.5,page=2]{subdivision}
    	\subcaption{}\label{fi:subdivision-2}
    \end{minipage}
    \caption{ (a)--(c) Illustration for the proof of Lemma~\ref{le:snodes}. (d)-(e) Removing subdivision vertices.}
\end{figure}

\begin{lemma}\label{le:snodes}
Let $\mu$ be an $S$-node of $T'$. Let $\gamma_{\nu_1}$, $\gamma_{\nu_2}$, $\dots$, $\gamma_{\nu_k}$ be the $k \geq 0$ visibility representations of the $k \geq 0$ children of $\mu$ (excluding $Q$-nodes) for which Invariants {\bf I1--I3} hold. Then $G_\mu$ admits a visibility representation $\gamma_\mu$ that respects Invariants {\bf I1--I3}, and such that each vertex $s_{\nu_i}=t_{\nu_{i-1}}$ ($i=2,\dots,k$) is  represented by at most four bars.
\end{lemma}
\begin{proof}
If $\mu$ has some $Q$-nodes among its children, we draw them as shown in Figure~\ref{fi:snodes-1}. The only exception is when one of them is incident to $t_\mu$, in which case it is drawn as in Figure~\ref{fi:snodes-2} so to avoid the addition of a horizontal bar to $t_\mu$. This ensures Invariants {\bf I1--I2}.

Then we merge the horizontal bar of $s_{\nu_i}=t_{\nu_{i-1}}$ in $\gamma_{\nu_i}$ with the vertical bar of the same vertex in $\gamma_{\nu_{i-1}}$, for $i=2,\dots,k$, as shown in Figure~\ref{fi:snodes-3}. This construction satisfies Invariant {\bf I3}. We remark that each vertex $s_{\nu_i}=t_{\nu_{i-1}}$ is now represented by at most four bars  (this fact will be used in the following), as desired.
\end{proof}

We now describe how to draw the edge $(s_\rho,t_\rho)$ represented by the root $\rho$ of $T'$. We observe that $s_\rho$ and $t_\rho$ are the poles of $\xi$, and hence they are represented in $\gamma_\xi$ by a single bar each. We represent $(s_\rho,t_\rho)$ with a horizontal visibility. To this aim, we add a vertical bar to $s_\rho$ whose bottommost point coincides with the leftmost point of $s_\rho$ and whose topmost point is slightly above the bottommost point of $t_\rho$ (so to ensure the $\varepsilon$-visibility model). The edge $(s_\rho,t_\rho)$ is now represented as a horizontal visibility between this vertical bar and the one representing $t_\rho$.

As a final step, we show how to remove the subdivision vertices used to subdivide some of the edges. Let $(u,v)$ be a subdivided edge corresponding to a $Q$-node of $T$ whose parent is a $P$-node $\mu$. Let $x$ be the subdivision vertex added to split $(u,v)$. As described in the proof on Lemma~\ref{le:snodes}, vertex $x$ is represented by a ``staircase'' of 4 bars. By turning either the visibility $(u,x)$, or the visibility $(v,x)$ into a bar, we get rid of $x$ and realize the edge $(u,v)$ as shown in Figures~\ref{fi:subdivision-1} and~\ref{fi:subdivision-2}. This operation increases  the number of bars used to represent $u$ or $v$ by one. Since the parent of $\mu$ is not an $R$-node, it can only be an $S$-node, and hence both $u$ and $v$ are represented by a set of bars of size at most 4 (Lemma~\ref{le:snodes}). However, both $u$ and $v$ can be incident to many subdivided edges. We now show how to ``charge'' at most two subdivided edges per vertex, meaning that the charged vertex will be the one taking the additional bar. As a consequence, each of these vertices is represented by at most six bars.

The idea is as follows. Let $\mu$ be an $R$-node of $T$ such that it does not have any $R$-node as descendant. Let $\nu_1,\dots,\nu_k$ be its children that are not $Q$-nodes. Since the subtree $T_{\nu_i}$ rooted at $\nu_i$ ($i=1,\dots,k$) does not contain any $R$-node, the pertinent graph $G_{\nu_i}$ is a partial 2-tree. We now show that every partial 2-tree admits an orientation of its edges such that every vertex has at most two outgoing edges. Since 2-trees are 2-degenerate, they can be made empty by iteratively removing a vertex $v$ with degree at most two. Orienting outwards the at most two  edges incident to $v$ leads to the desired orientation. Observe that the penultimate vertex only has one outgoing edge $e$. It is not difficult to see that the order of the removed vertices can be chosen so that the last edge $e$ is a predefined one. By applying this procedure on $G_{\nu_i}$ and choosing $e=(s_{\nu_i},t_{\nu_i})$ as predefined edge, we have that all vertices of $G_{\nu_i}$ have at most two outgoing edges, except for $s_{\nu_i}$ and $t_{\nu_i}$. Hence, the number of bars used to represent $s_{\nu_i}$ and $t_{\nu_i}$ does not increase (recall that $e=(s_{\nu_i},t_{\nu_i})$ is already drawn in $\gamma_\mu$), while the number of bars used to represent any other vertex in $G_{\nu_i}$ is at most 6. Apply the above procedure for all $R$-nodes that do not have any $R$-node as descendant. Afterwards, prune down the subtrees of $T$ rooted at such $R$-nodes and iterate the procedure until there are no more $R$-nodes. Repeat the procedure on the remaining tree.

This algorithm can be implemented to run in $O(n)$ time by avoiding the direct computation of an \opvr, and by storing instead the information required to compute an orthogonal representation $H$ of $G$. This, in particular, removes the scaling operations that only affect the length of the edges. By applying the compaction step of the TSM framework, we finally obtain an \opvr of $G$ in $O(n)$ time on an integer grid of size $O(n) \times O(n)$. Since every vertex is composed of at most $12$ bars in the visibility representation, in the final \opvr it is drawn as an orthogonal polygon with at most $(12-1)\cdot 2$ reflex corners.  This proves Theorem~\ref{th:2conn1plane-ub}.

\section{Implementation and Experiments}\label{se:experiments}
We implemented the optimization algorithm of Theorem~\ref{th:test-opt} in C++, using the GDToolkit library~\cite{gdt-13}. To evaluate the performance of the algorithm in practice, we tested it on a large set of 1-plane graphs, which always admit an \opvr (Theorem~\ref{th:1-planar}). Other than evaluating the running time of the algorithm, we have the two following objectives:

\begin{description}
\item[\textbf{Obj-1.}]  Measuring the vertex complexity of the computed {\opvr}s. In particular, for 3-connected 1-plane graphs the gap between the upper bound of 12 and the lower bound of 2 is intriguing. We expect that in practice the vertex complexity is closer to the lower bound, since the algorithm behind our upper bound imposes strong restrictions on the computed {\opvr}s. For instance, it assumes that crossing-free edges are always drawn as vertical bars, which might not be the case for an optimal solution. 

\item[\textbf{Obj-2.}] Establishing ``how much'' the computed drawings look like rectangle visibility representations, independently of their vertex complexity. To this aim, for every computed \opvr with vertex complexity $k$, we measure the percentage of vertices whose corresponding polygons have $i$ reflex corners, for any integer $i \in [0, \dots, k]$. We recall that our optimization algorithm not only minimizes the vertex complexity, but within all the optimal solutions it computes one having the minimum number of reflex corners (see Theorem~\ref{th:test-opt}). Thus, we always expect a high number of vertices represented with low vertex complexity (ideally as rectangles).
\end{description} 

\begin{figure}[tb]
 \centering
 \includegraphics[scale=0.3]{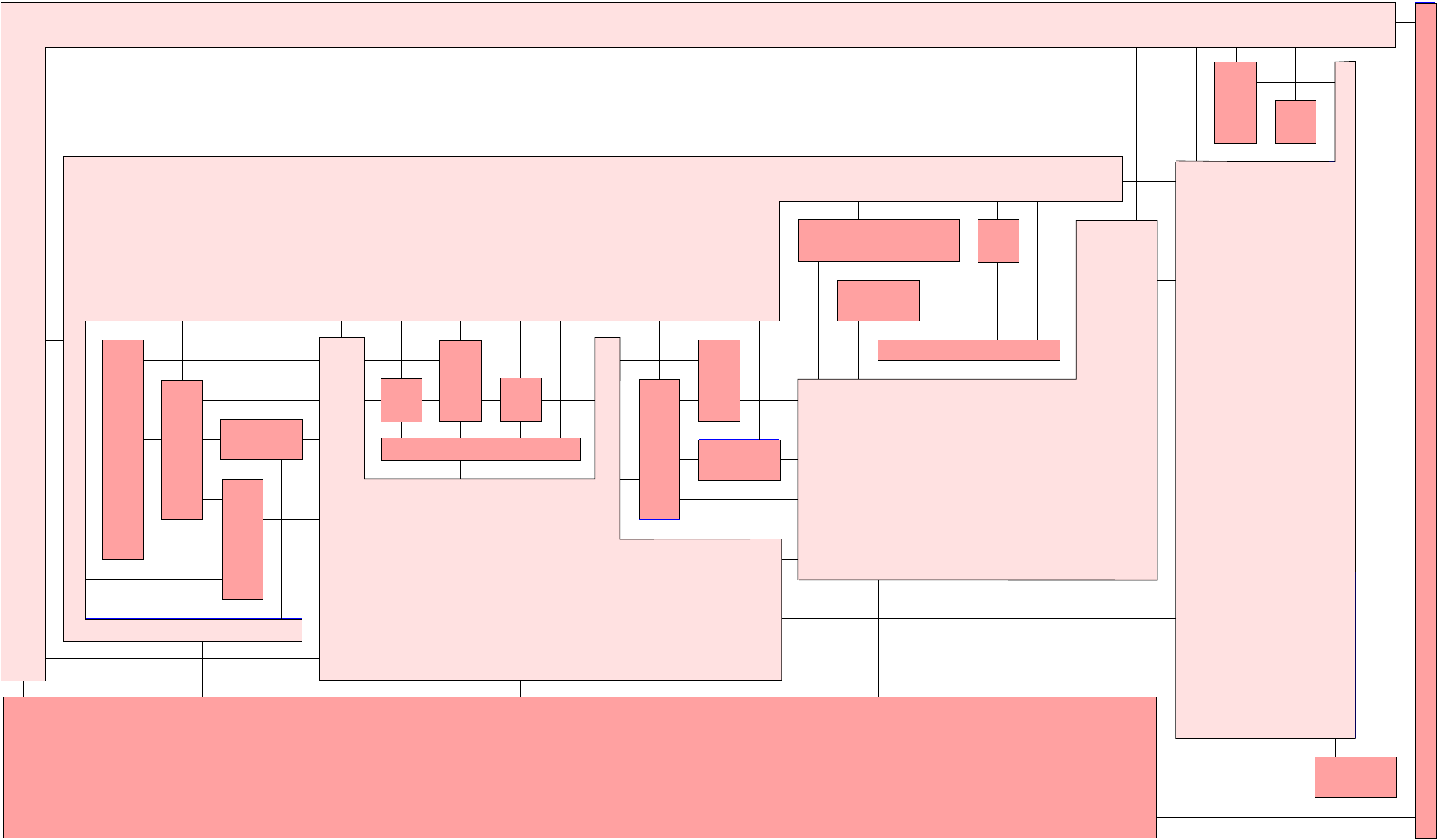}
 \caption{An \opvr with 25 vertices and vertex complexity 3. Rectangular vertices have a darker color.}\label{fi:drawing-example}
\end{figure}

\paragraph{Test suite.}
We generated three different subsets of (simple) 1-plane graphs, which we call {\tt GEN}, {\tt BIC}, and {\tt TRIC}, respectively. Each subset consists of 170 graphs (thus 510 instances in total). The number of vertices of each graph ranges from 20 to 100. The graphs in {\tt GEN} are general 1-plane graphs, while those in {\tt BIC} and in {\tt TRIC} are always 2-connected and 3-connected, respectively. All graphs are maximal, which means that no further edges can be added in their embedding while preserving 1-planarity. Clearly, augmenting a 1-plane graph to a maximal one cannot lower the vertex complexity of its {\opvr}s, whereas it increases the running time required to compute a solution due to the increased number of edges. 

 \begin{figure}[t]
    \centering
    \begin{minipage}[b]{.49\textwidth}
    	\centering
    	\includegraphics[scale=0.34]{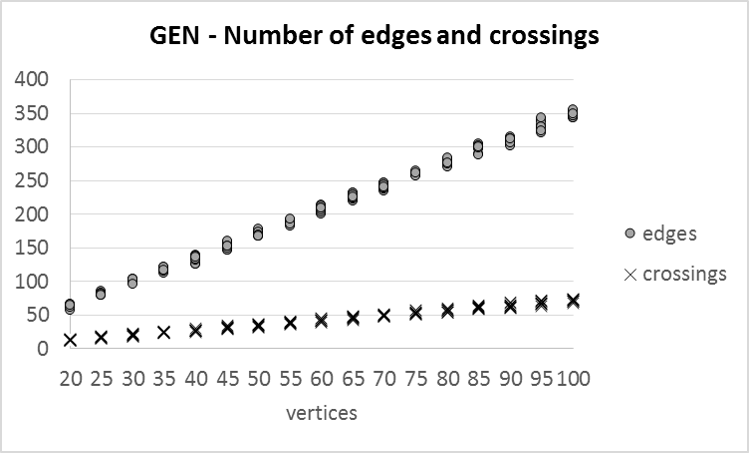}
    	\subcaption{{\tt GEN} - Number of edges and crossings.}\label{ch:chart-GEN-EdgesCrossings}
    \end{minipage}
    \begin{minipage}[b]{.49\textwidth}
    	\centering
    	\includegraphics[scale=0.34]{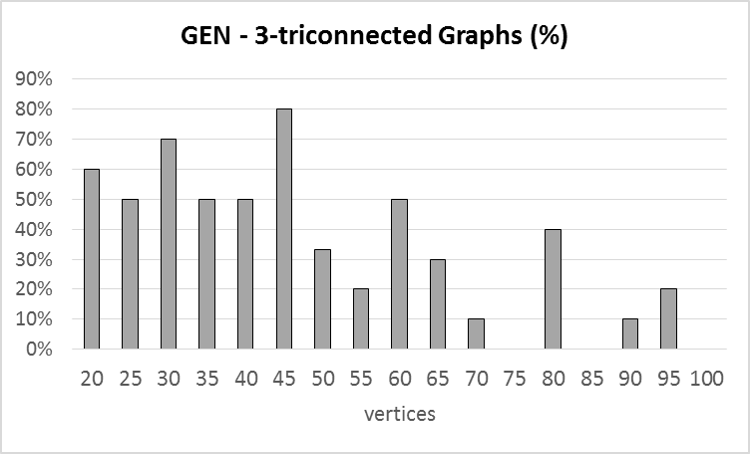}
    	\subcaption{{\tt GEN} - Average \% of 3-connected graphs.}\label{ch:chart-GEN-TricGraphs}
    \end{minipage}
    \begin{minipage}[b]{.49\textwidth}
    	\centering
    	\includegraphics[scale=0.34]{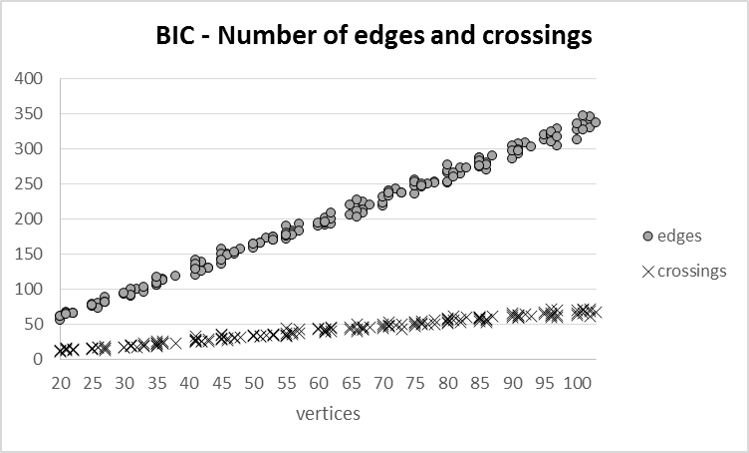}
    	\subcaption{{\tt BIC} - Number of edges and crossings.}\label{ch:chart-BIC-EdgesCrossings}
    \end{minipage}   
    \begin{minipage}[b]{.49\textwidth}
    	\centering
    	\includegraphics[scale=0.34]{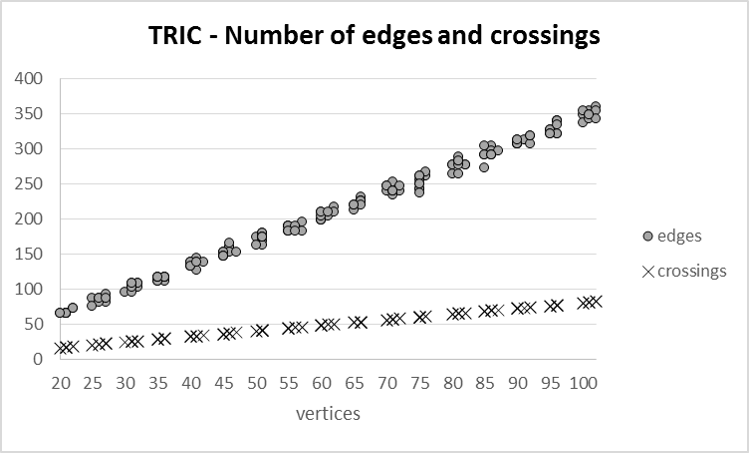}
    	\subcaption{{\tt TRIC} - Number of edges and crossings.}\label{ch:chart-TRIC-EdgesCrossings}
    \end{minipage}  
    \caption{Charts for the experimental data. The measured values are reported with dependence on the number of vertices ($x$-axis).}\label{ch:additional}
\end{figure} 

Each graph with $n$ vertices in {\tt GEN} is obtained as follows. We first randomly generate an $n$-vertex 2-connected plane graph with the algorithm described in~\cite{DBLP:journals/tc/BertolazziBD00}. We then add as many edges as possible such that each new edge crosses a previously uncrossed edge of the graph and no multiple edge is introduced. We finally add a random sequence of uncrossed edges to get maximality (that is, no further edge can be added without either violating 1-planarity or introducing a multiple edge). Although in principle every maximal 1-plane graph can be generated with this approach, we observed that in practice all instances in {\tt GEN} admitted an \opvr with vertex complexity at most one (see the results below). Hence, we generated the sets {\tt BIC} and {\tt TRIC}, which contain more difficult instances, obtained by explicitly adding the 1-plane configurations used to prove our lower bounds. For a given positive integer $n$, a graph in {\tt BIC} is generated as follows: $(i)$ start from a randomly generated $k$-vertex 2-connected plane graph, where $k$ is a fraction of $n$ (we chose $k = 0.2$, as we observed that larger values give rise to graphs whose {\opvr}s have smaller vertex complexity); $(ii)$ perform a random sequence of operations, where each operation adds an augmented B-, or W-, or T-configuration, or a new crossing edge, or a new pair of crossing edges to the graph, until the number of vertices reaches or exceeds $n$ (multiple edges are not added); $(iii)$ add a final random sequence of uncrossed edges to get maximality. With this approach the resulting graph might have a number of vertices slightly larger than $n$ (at most $n+3$). The graphs in {\tt TRIC} are generated analogously, but with the following two variants, which are needed to keep the graphs 3-connected: $(i)$ the initial 2-connected graph is randomly triangulated before adding 1-plane configurations; $(ii)$ no W-configuration is added, and each augmented B-configuration is added only if it is possible to connect one of its internal vertices to the rest of the graph using an additional crossing edge. 

The average density of the {\tt GEN} graphs is 3.4 and, on average, 41.2\% of their edges are crossing edges: the variance for these two parameters is very low. About 33.7\% of these graphs are 3-connected (see Figures~\ref{ch:chart-GEN-EdgesCrossings} and~\ref{ch:chart-GEN-TricGraphs}). The {\tt BIC} and {\tt TRIC} graphs have an average density similar to that of the {\tt GEN} graphs: 3.2 for {\tt BIC} and 3.4 for {\tt TRIC} (Figures~\ref{ch:chart-BIC-EdgesCrossings} and~\ref{ch:chart-TRIC-EdgesCrossings}). The percentage of crossing edges in the {\tt BIC} graphs is very close to that of the {\tt GEN} graphs, while for the {\tt TRIC} graphs it is slightly higher (47.8\% on  average).

\begin{figure}[t!]
    \centering
    \begin{minipage}[b]{.49\textwidth}
    	\centering
    	\includegraphics[width=0.9\textwidth]{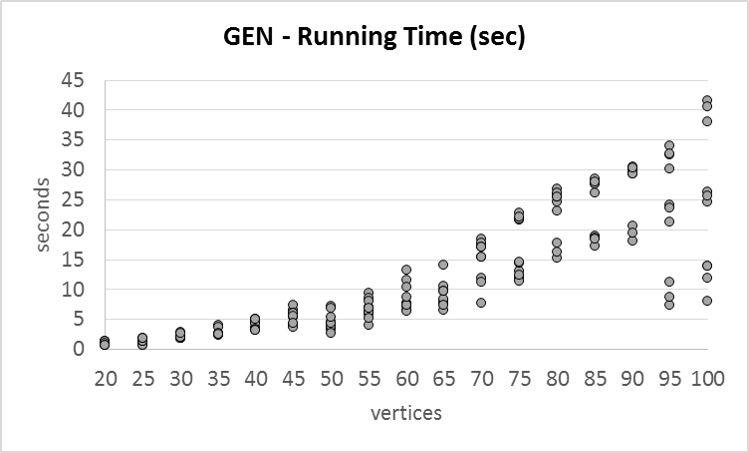}
    	\subcaption{Running time.}\label{ch:chart-GEN-Runtime}
    \end{minipage}
    \begin{minipage}[b]{.49\textwidth}
    	\centering
    	\includegraphics[width=0.9\textwidth]{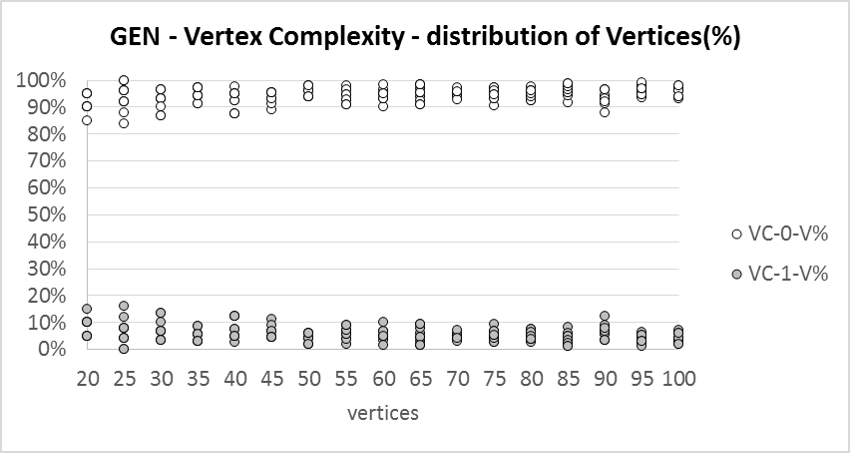}
    	\subcaption{\% of vertices with complexity $i$ (VC-$i$-V\%).}\label{ch:chart-GEN-VertexComplexityDistrV}
    \end{minipage}    
    \begin{minipage}[b]{.49\textwidth}
    	\centering
    	\includegraphics[width=0.9\textwidth]{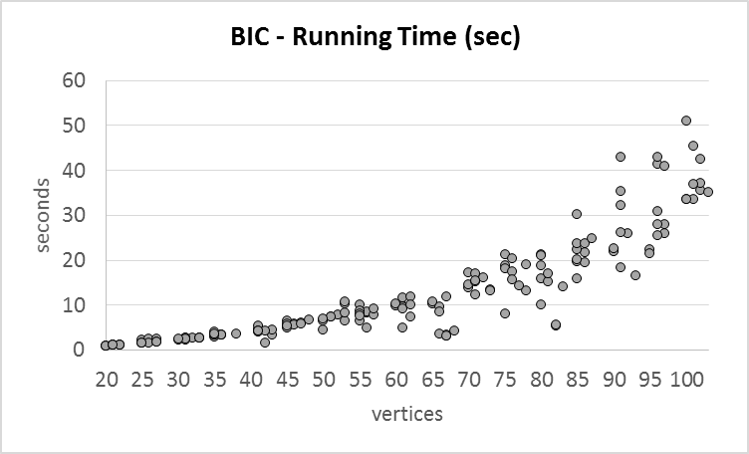}
    	\subcaption{Running time.}\label{ch:chart-BIC-Runtime}
    \end{minipage}
   \begin{minipage}[b]{.49\textwidth}
    	\includegraphics[width=0.9\textwidth]{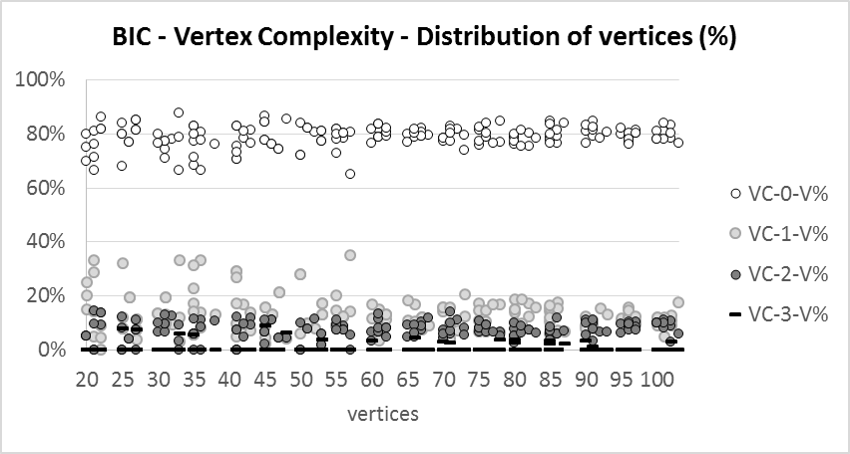}
    	\subcaption{\% of vertices with complexity $i$ (VC-$i$-V\%).}\label{ch:chart-BIC-VertexComplexityDistrV}
    \end{minipage}     
    \begin{minipage}[b]{.49\textwidth}
    	\centering
    	\includegraphics[width=0.9\textwidth]{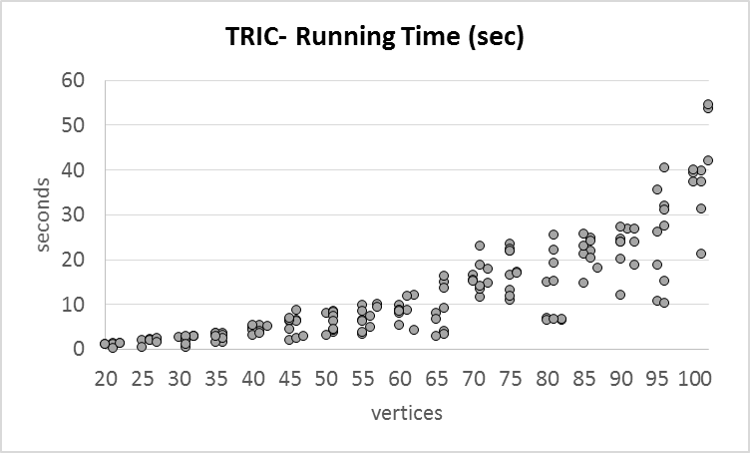}
    	\subcaption{Running time.}\label{ch:chart-TRIC-Runtime}
    \end{minipage}
    \begin{minipage}[b]{.49\textwidth}
    	\includegraphics[width=0.9\textwidth]{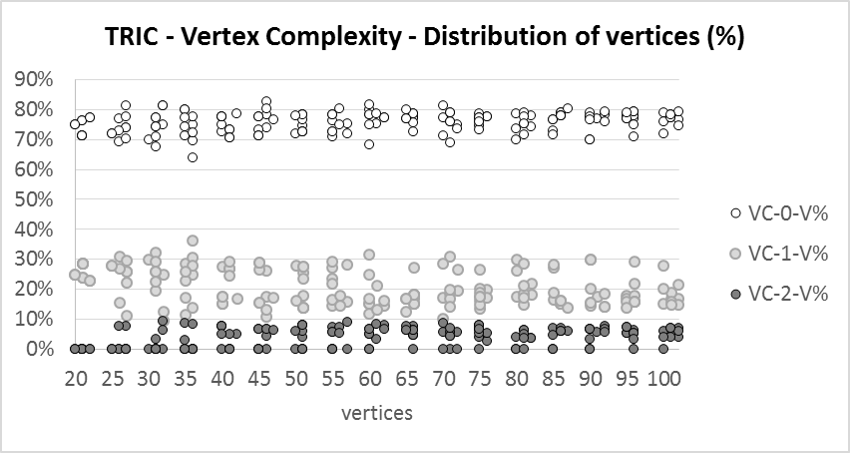}
    	\subcaption{\% of vertices with complexity $i$ (VC-$i$-V\%).}\label{ch:chart-TRIC-VertexComplexityDistrV}
    \end{minipage}
     \begin{minipage}[b]{.49\textwidth}
    	\centering
    	\includegraphics[width=0.9\textwidth]{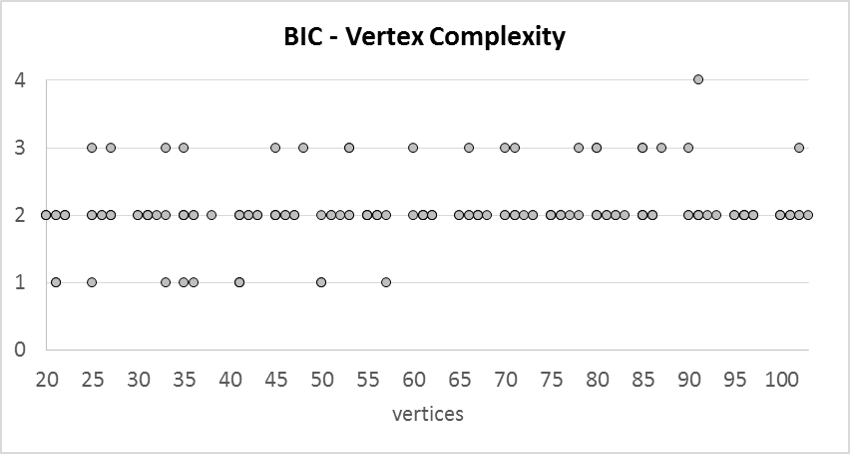}
    	\subcaption{{\tt} BIC - Vertex complexity.}\label{ch:chart-BIC-VertexComplexity}
    \end{minipage}
    \begin{minipage}[b]{.49\textwidth}
    	\centering
    	\includegraphics[width=0.9\textwidth]{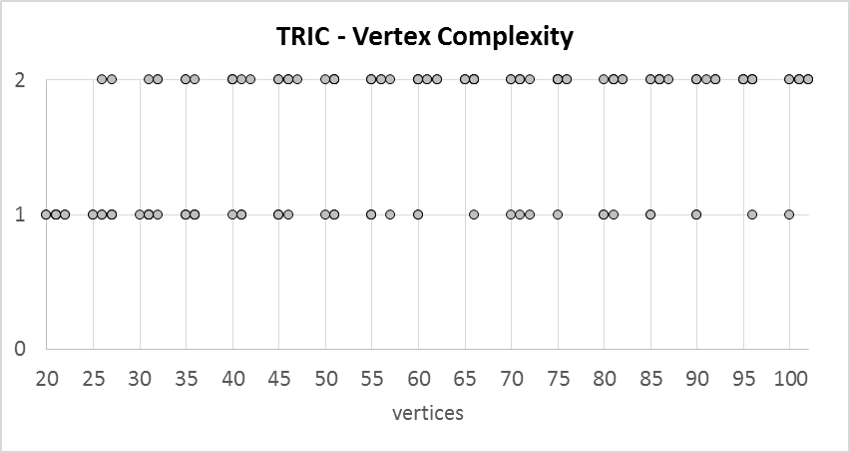}
    	\subcaption{{\tt} TRIC - Vertex complexity.}\label{ch:chart-TRIC-VertexComplexity}
    \end{minipage}        
    \caption{Charts (scatter plots) that summarize some experimental data. The measured values are reported with dependence on the number of vertices ($x$-axis).}\label{ch:all}
\end{figure}

\paragraph{Results.}
The computations have been executed on a common laptop, equipped with an Intel I7 processor and 8 GB of RAM. The software ran in the Oracle VirtualBox environment, under the Linux Ubuntu OS. For the {\tt GEN} graphs, the optimization drawing algorithm took less than 15 seconds for all instances up to 60 vertices, and about 41 seconds on the largest instance, having 100 vertices and 355 edges (Figure~\ref{ch:chart-GEN-Runtime}). Concerning the vertex complexity, the optimal solutions of all {\tt GEN} graphs required only vertex complexity 1, except two of them that had an \opvr with vertex complexity 0. Figure~\ref{ch:chart-GEN-VertexComplexityDistrV} shows, for each instance, the percentage of vertices with 0 (i.e., rectangular vertices) and with one reflex corner: the percentage of vertices drawn as rectangles is around 90\%, and more than 80\% for every instance. Hence, a big portion of each drawing looks like a rectangle visibility representation.

The running times for {\tt BIC} and {\tt TRIC} reflect the behavior observed for the {\tt GEN} graphs. However, the largest instances of {\tt BIC} and {\tt TRIC} often appear to be computationally more expensive (Figures~\ref{ch:chart-BIC-Runtime} and~\ref{ch:chart-TRIC-Runtime}). The vertex complexity required by the different instances is shown in Figures~\ref{ch:chart-BIC-VertexComplexity} and~\ref{ch:chart-TRIC-VertexComplexity}. Every instance of {\tt TRIC} admitted a drawing with vertex complexity either 1 (37.65\% of the instances) or 2 (62.35\%), while the {\tt BIC} graphs also required vertex complexity 3 (11.76\% of the instances) and, in one case, vertex complexity 4; however, the majority of the instances (80.59\%) required vertex complexity 2.
For each instance, the distribution of the number of vertices drawn with $i$ reflex corners, where $i$ ranges from 0 to the vertex complexity required by that instance, is depicted in Figures~\ref{ch:chart-BIC-VertexComplexityDistrV} and~\ref{ch:chart-TRIC-VertexComplexityDistrV}. To avoid visual clutter,  we did not report the data about the unique drawing with vertex complexity 4 in the chart of Figure~\ref{ch:chart-BIC-VertexComplexityDistrV}; this drawing has only two vertices with 4 reflex corners (from a total of 91 vertices). From the charts, one can see that the percentage of vertices drawn as rectangles is still very high (around 80\% for {\tt BIC} and around 75\% for {\tt TRIC}). 

Overall, the experimental results confirmed our expectations about {\bf Obj-1} and {\bf Obj-2}. 
An example of an \opvr computed with our algorithm is depicted in Fig~\ref{fi:drawing-example}.

\section{Conclusions and Open Problems}\label{se:conclusions}

In this paper we have introduced the notion of ortho-polygon visibility representations (\opvr{s}), a generalization of rectangle visibility representations where vertices can be represented as orthogonal polygons instead of rectangles. We have provided a quadratic-time algorithm that tests embedded graphs for representability and, if the test is affirmative, it computes an embedding-preserving \opvr with minimum vertex complexity, i.e., with the minimum number of reflex corners per vertex. Motivated by recent results on rectangle visibility representations~\cite{SoCG}, we have studied \opvr{s} of 1-planar graphs. We have shown that for 3-connected 1-plane graphs an \opvr with vertex complexity at most 12 can be computed in linear-time. We also showed that the vertex complexity of an \opvr of a 3-connected 1-plane graph is at least 2 for some instances. For 2-connected 1-plane graphs, the vertex complexity is $\Omega(n)$ for some instances, but if the graphs do not have W-configurations, a 1-plane embedding that guarantees constant vertex complexity can be constructed in $O(n)$ time. Finally, we ran an experimental study to estimate the vertex complexity of \opvr{s} of 1-plane graphs in practice.

The results in this paper naturally raise interesting open problems. Among them are: 

\begin{enumerate}
\item Close the gap between the upper bound and the lower bound on the vertex complexity of {\opvr}s of 3-connected 1-plane graphs (see Theorems~\ref{th:3conn-ub} and \ref{th:3conn-lb}). 
\item As shown in Section~\ref{se:experiments}, many vertices in an optimal \opvr are rectangles in practice. We find it interesting to study the problem of computing {\opvr}s that maximize the number of rectangular vertices, even at the expense of sub-optimal vertex complexity. 
\item Theorem~\ref{th:2conn1plane-ub} constructs 1-planar embeddings that guarantee constant vertex complexity if the input does not have W-configurations. What 2-connected 1-plane graphs admit a 1-planar {\opvr} with constant vertex complexity? 
\end{enumerate}

{\small \bibliography{paper}}
\bibliographystyle{abbrv}

\appendix

\clearpage

\renewcommand\thesection{\Alph{section}}

\section{Orthogonal Representations and Network Flow Model}\label{ap:orthorep-flow}

In this section, we recall basic definitions and main results related to the problem of computing orthogonal representations exploiting the network flow model by Tamassia. We refer the reader to~\cite{dett-gdavg-99,Garg1997} for further details.

Let $G$ be a plane graph (possibly with multiple edges and self-loops) whose maximum vertex degree is four. Let $e=(u,v)$ be an edge of $G$. The two possible orientations $(u,v)$ and $(v,u)$ of $e$ are called \emph{darts}. A dart is said to be \emph{counterclockwise} with respect to face $f$ if $f$ is on the left hand side when walking along the dart according to its orientation. Denote by $D(u)$ the set of darts exiting from $u$ and by $D(f)$ the set of counterclockwise darts with respect to $f$. 

An \emph{orthogonal representation} of $G$ is an assignment to each dart $(u,v)$ of two values $\alpha(u,v) \in \{1,2,3,4\}$ and $\beta(u,v) \in \mathbb{N}$  that satisfies the following conditions.

\begin{description}
\item[\bf C1.] $1 \leq \alpha(u,v) \leq 4$;
\item[\bf C2.] $\beta(u,v) \geq 0$;
\item[\bf C3.] $\sum_{(u,v) \in D(u)}\alpha(u,v)=4$;
\item[\bf C4.] For each internal face $f$: $\sum_{(u,v) \in D(f)}(\alpha(u,v)+\beta(v,u)-\beta(u,v))=2\deg(f)-4$;
\item[\bf C5.] For the outer face $f_{ext}$: $\sum_{(u,v) \in D(f_{ext})}(\alpha(u,v)+\beta(v,u)-\beta(u,v))=2\deg(f)+4$.
\end{description}

The value $\alpha(u,v)\cdot \ph$ is the angle that dart $(u,v)$ forms with the dart following it in the circular counterclockwise order around $u$, while the value $\beta(u,v)$ is the number of bends of $\ph$ along the dart $(u,v)$. Condition {\bf C1} expresses the fact that the sum of angles around each vertex is $2\p$, while  {\bf C2} (respectively {\bf C3}) expresses the fact that the sum of the angles at the vertices and bends of an internal face (respectively outer) is equal to $\pi (p-2)$ (respectively $\pi (p+2)$), where $p$ is the number of such angles.

\medskip

An orthogonal representation of $G$ with the minimum number of bends can be computed by means of a flow network $N$. In the flow network $N$, each unit of flow corresponds to a $\ph$ angle, each vertex supplies 4 units of flow, and each face consumes an amount of flow proportional to its degree. Bends along edges correspond to unit of flows transferred across adjacent faces, and each bend has a unit cost in the network. The flow network $N$ is constructed as follows. The nodes of network $N$ are the vertices and faces of $G$.  Each \emph{vertex-node} $v$ supplies $\sigma(v)=4$ units of flow, and each \emph{face-node} $f$ consumes $\tau(f)$ units of flow, where
\[
\tau(f) =   \left\{ \begin{array}{l l}
        2 \deg(f) - 4  & \mbox{if $f$ is an internal face}\\
        2 \deg(f) + 4  & \mbox{if $f$ is the outer face.}
        \end{array} \right.
\]
By Euler's formula, $\sum_v \sigma(v) = \sum_f \tau(f)$, i.e., the total supply is equal to the total consumption.

\noindent For each dart $(u,v)$ of $G$, with faces $f$ and $g$ on its left and right, respectively, $N$ has two arcs:
\begin{itemize}
    \item an arc $(u,f)$ with lower bound $\lambda(v,f)=1$, capacity $\mu(v,f)=4$, and cost $\chi(v,f)=0$;
    \item an arc $(f,g)$ with lower bound $\lambda(v,f)=0$, capacity $\mu(v,f)=+\infty$, and cost $\chi(v,f)=1$;
\end{itemize}

The conservation of flow at the vertices expresses the fact that the sum of the angles around a vertex is equal to $2 \pi$.  The conservation of flow at the faces expresses the fact that the sum of the angles at the vertices and bends of an internal face is equal to $\pi (p-2)$, where $p$ is the number of such angles.  For the outer face, the above sum is equal to $\pi(p+2)$.

It can be shown that every feasible flow $\phi$ in network $N$ corresponds to an admissible orthogonal representation for graph $G$, whose number of bends is equal to the cost of flow $\phi$.  Namely, let $\Phi$ be a flow of $N$ with cost $b$. Then, for each dart $(u,v)$  whose associated arcs of $N$ are $(u,f)$ and $(f,g)$, we set $\alpha(u,v)=\Phi(u,f)$ and $\beta(u,v) = \Phi(f,g)$. On the other hand, by just setting $\Phi(u,f)=\alpha(u,v)$ and $\Phi(f,g)=\beta(u,v)$, an orthogonal representation $H$ with at most $b$ bends is transformed into a feasible flow $\Phi$ of $N$ with cost $b$. Hence, the following theorem summarizes the above discussion.

\begin{theorem}[see e.g.~\cite{dett-gdavg-99}]
Let $G$ be a plane graph with $n$ vertices and maximum vertex degree four. An orthogonal representation $H$ of $G$ with the minimum number of bends can be computed in $O(T(n))$ time, where $T(n)$ is the time for computing a min-cost flow of the flow network $N$ associated with $G$.
\end{theorem}

\section{The $SPQR$-tree Decomposition}\label{ap:spqrtree}

The following definitions and observations are useful for the proof of Theorem~\ref{th:2conn1plane-ub}.

Let $G$ be a 2-connected graph.
A \emph{separation pair} is a pair of vertices whose removal disconnects $G$.
A \emph{split pair} is either a separation pair or a pair of adjacent vertices.
A \emph{split component} of a split pair $\{u,v\}$ is either an edge $(u,v)$ or a maximal subgraph $G_{uv} \subset G$ such that $\{u,v\}$ is not a split pair of $G_{uv}$.
Vertices $\{u,v\}$ are the \emph{poles} of $G_{uv}$.
The \emph{$SPQR$-tree} $T$ of $G$ with respect to an edge $e$ is a rooted tree that describes a recursive decomposition of $G$ induced by its split pairs~\cite{DBLP:journals/siamcomp/BattistaT96}.
In what follows, we call \emph{nodes} the vertices of $T$, to distinguish them from the vertices of $G$.
The nodes of $T$ are of four types $S$,$P$,$Q$, or $R$.
Each node $\mu$ of $T$ has an associated 2-connected multigraph called the \emph{skeleton of $\mu$} and denoted by $sk(\mu)$.
At each step, given the current split component $G^*$, its split pair $\{s,t\}$, and a node $\nu$ in $T$, the node $\mu$ of the tree corresponding to $G^*$ is introduced and attached to its parent vertex $\nu$, while the decomposition possibly recurs on some split component of $G^*$.
At the beginning of the decomposition the parent of $\mu$ is a $Q$-node corresponding to $e=(u,v)$, $G^* = G \setminus e$, and $\{s,t\} = \{u,v\}$.

\begin{figure}
\centering
\begin{minipage}[b]{.3\textwidth}
    	\centering
    	\includegraphics[scale=0.4, page=1]{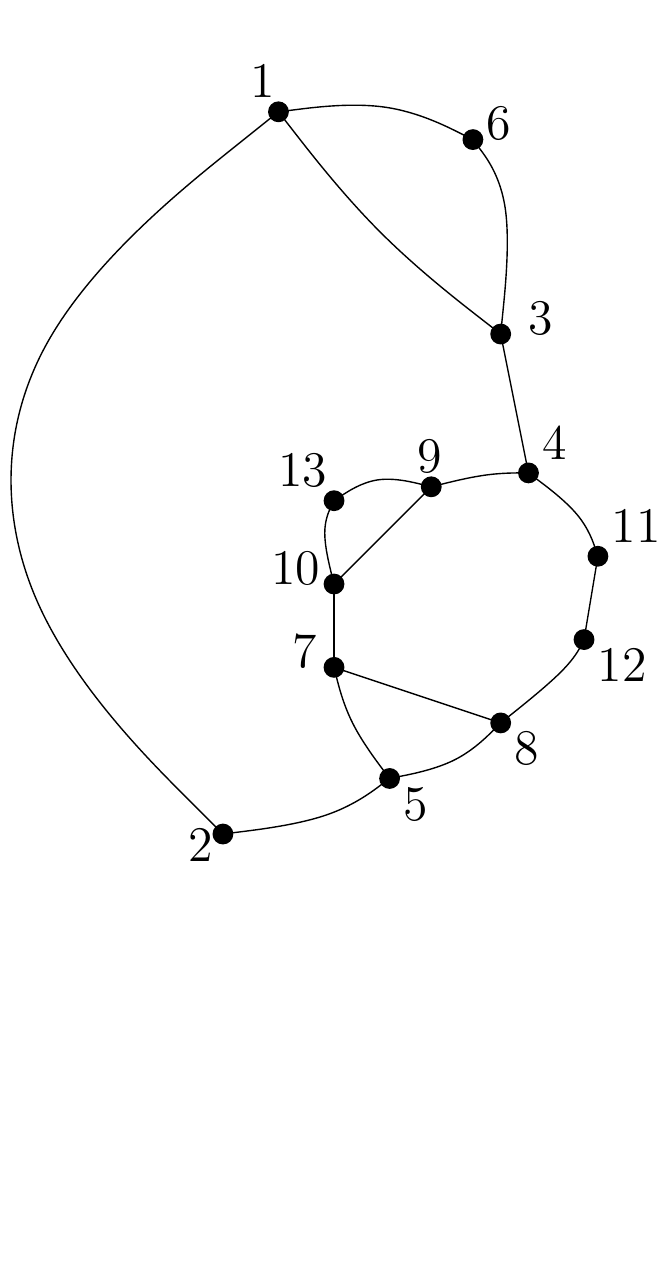}
    	\subcaption{}\label{fi:SPQR-example.1}
    \end{minipage}
    \begin{minipage}[b]{.65\textwidth}
    	\centering
    	\includegraphics[scale=0.5, page=2]{figures/SPQR-example}
    	\subcaption{}\label{fi:SPQR-example.2}
    \end{minipage}
\caption{\label{fi:SPQR-example}(a) A graph $G$; (b) The $SPQR$-tree $T$ of $G$. For each node that is not a $Q$-tree the skeleton is depicted in the gray balloons; for $Q$-nodes the corresponding edge is shown.}
\end{figure}

\textbf{Base case}: $G^*$ consists of a single edge between $s$ and $t$. Then, $\mu$ is a $Q$-node whose skeleton is $G^*$ itself plus the reference edge between $s$ and $t$.

\textbf{Parallel case}: The split pair $\{s,t\}$ has $G_1,\dots,G_k$ ($k \geq 2$) split components.
Then, $\mu$ is a $P$-node whose skeleton is a set of $k+1$ parallel edges between $s$ and $t$, one for each split component $G_i$ plus the reference edge between $s$ and $t$.
The decomposition recurs on $G_1,\dots,G_k$ with $\mu$ as parent node.

\textbf{Series case}: $G^*$ is not 2-connected and it has at least one cut vertex (a vertex whose removal disconnects $G^*$).
Then, $\mu$ is an $S$-node whose skeleton is defined as follows.
Let $v_1,\dots,v_{k-1}$, where $k \geq 2$, be the cut vertices of $G^*$.
The skeleton of $\mu$ is a path $e_1,\dots,e_k$, where $e_i= (v_{i-1},v_i)$, $v_0=s$ and $v_k=t$, plus the reference edge between $s$ and $t$ which makes the path a cycle.
The decomposition recurs on the split components corresponding to each $e_1,\dots,e_k$ with $\mu$ as parent node.

\textbf{Rigid case}: None of the other cases is applicable.
A split pair $\{s',t'\}$ is maximal with respect to $\{s,t\}$, if for every other split pair $\{s^*,t^*\}$, there is a split component that includes the vertices $s',t',s,t$.
Let $\{s_1,t_1\},\dots,\{s_k,t_k\}$ be the maximal split pairs of $G^*$ with respect to $\{s,t\}$ ($k \geq 1$), and, for $i=1,\dots,k$, let $G_i$ be the union of all the split components of $\{s_i,t_i\}$.
Then $\mu$ is an $R$-node whose skeleton is obtained from $G^*$ by replacing each component $G_i$ with an edge between $s_i$ and $t_i$, plus the reference edge $(s,t)$.
The decomposition recurs on each $G_i$ with $\mu$ as parent node.

\medskip

Figure~\ref{fi:SPQR-example} shows a graph and its $SPQR$-tree. For each node that is not a $Q$-tree the skeleton is depicted; for $Q$-nodes the corresponding edge is shown. The $SPQR$-tree $T$ of a graph $G$ with $n$ vertices and $m$ edges has $m$ $Q$-nodes and $O(n)$ $S$-, $P$-, and $R$-nodes. Also, the total number of vertices of the skeletons stored at the nodes of $T$ is $O(n)$.

\medskip

If $G$ is an embedded graph, then each pertinent graph $G_\mu$ of a node $\mu$ of $T$ is also an embedded graph.  Furthermore, the skeleton $sk(\mu)$ of $\mu$ inherits an embedding from $G_\mu$. For our purposes, we observe that if $G$ is a 1-plane graph, then the skeleton of an $R$-node is also a 1-plane graph. Moreover, we remark that the skeleton of an $R$-node is $3$-connected by definition.

The $SPQR$-tree can also be exploited to modify the embedding of $G$. A  split component can be \emph{flipped} around its poles, hence reversing the order of the edges of the split component around its poles.  A \emph{swap} operation consists of exchanging the position of two split components of the same split pair. If $G$ is 1-plane, both these operations modify the embedding of $G$ without introducing additional crossings, and thus preserve 1-planarity.

\end{document}